\documentclass[a4paper, 11pt]{article}
\usepackage{fullpage}

\usepackage{t1enc}
\usepackage[utf8]{inputenc}

\usepackage{latexsym,amssymb,amsmath,amsthm,graphicx,enumerate,mathtools}
\usepackage{cite,comment,easy-todo}
\usepackage{hyperref}

\newcommand{\maxwind}{maximum weight common independent set}
\newcommand{\cind}{common independent set}

\newcommand{\cI}{\ensuremath{\mathcal I}}

\newcommand{\kms}{\mathcal{N}}
\newcommand{\nearoptone}{{\sc binary-pop-near-opt}}
\newcommand{\nearopt}{{\sc pop-near-opt}}
\newcommand{\exm}{{\sc exact-matching}}
\newcommand{\Mk}{{\mathcal{M}_{k'}}}
\newcommand{\Nk}{{\mathcal{N}_{k}}}
\newcommand{\Ncal}{{\mathcal{N}}}
\newcommand{\Eblue}{E'_{\sf blue}}
\newcommand{\Ered}{E'_{\sf red}}

\newcommand{\exmw}{{\sc opt-exact-matching}}
\newcommand{\kpm}{\textsc{unary-pop-near-opt}}
\newcommand{\expm}{\textsc{exact-size-pop}}
\newcommand{\kpop}{$k$-popular matching}

\theoremstyle{plain}

\newtheorem{thm}{Theorem}
\newtheorem{theorem}[thm]{Theorem}

\newtheorem{corollary}[thm]{Corollary}

\newtheorem{lemma}[thm]{Lemma}

\newtheorem{claim}[thm]{Claim}

\theoremstyle{definition}
\newtheorem*{definition}{Definition}
\newtheorem{remark}[thm]{Remark}

\newtheorem*{problem}{Problem}

\usepackage{braket}
\newcommand{\vote}{{\sf vote}}

\newcommand{\score}{{\sf score}}
\newcommand{\opt}{{\sf opt}}

\newcommand{\Z}{\mathbb{Z}}
\newcommand{\Zp}{{\mathbb{Z}}_+}
\newcommand{\R}{\mathbb{R}}

\newcommand{\lev}{{\sf lev}}

\newcommand{\I}{\mathcal{I}}
\newcommand{\J}{\mathcal{J}}
\newcommand{\C}{\mathcal{C}}

\newcommand{\B}{\mathcal{B}}

\title{Popular Maximum-Utility Matchings with Matroid Constraints} 
\author{
Gergely Csáji\thanks{MTA-ELTE Momentum Matroid Optimization Research Group, Department of Operations Research, Eötvös Loránd University, Budapest, Hungary. Email: \texttt{csajigergely@student.elte.hu} and HUN-REN Centre for Economic and Regional Studies, Hungary.  Email: \texttt{csaji.gergely@krtk.hun-ren.hu}}\and
Tamás Király\thanks{HUN-REN--ELTE Egerv\'ary Research Group, Department of Operations Research, E\"otv\"os Loránd University, Budapest, Hungary. Email: \texttt{tamas.kiraly@ttk.elte.hu}}
\and
Kenjiro Takazawa%
\thanks{Department of Industrial and Systems Engineering, Faculty of Science and Engineering, Hosei University, Tokyo 184-8584, Japan. Email: \texttt{takazawa@hosei.ac.jp}}\and 
Yu Yokoi%
\thanks{Department of Mathematical and Computing Science, School of Computing, Tokyo Institute of Technology, Tokyo 152-8552, Japan. Email: \texttt{yokoi@c.titech.ac.jp}}
}
\begin{document}
\maketitle

\begin{abstract}
We investigate weighted settings of popular matching problems with matroid constraints. 
The concept of {\em popularity} was originally defined for matchings in bipartite graphs, where vertices have preferences over the incident edges. There are two standard models depending on whether vertices on one or both sides have preferences. A matching $M$ is popular if it does not lose a head-to-head election against any other matching. 
In our generalized models, one or both sides have matroid constraints, and a weight function is defined on the ground set. 
Our objective is to find a popular optimal matching, i.e., a maximum-weight matching that is popular among all maximum-weight matchings satisfying the matroid constraints.
For both one- and two-sided preferences models, we provide efficient algorithms to find such solutions, combining algorithms for unweighted models with fundamental techniques from combinatorial optimization. 
The algorithm  for the one-sided preferences model is further extended to a model where the weight function is generalized to an M$^\natural$-concave utility function.
Finally, we complement these tractability results by providing hardness results for the problems of finding a popular near-optimal matching. These hardness results hold even without matroid constraints and with very restricted weight functions.
\end{abstract}

\section{Introduction}
The study of \emph{popular matchings} is a relatively new topic at the intersection of algorithmic game theory, operations research, and economics. It examines the concept of weak Condorcet winner \cite{nicolas1785essai, condorcet} in the context of matching under preferences. A matching is called {\em popular} if it does not lose a head-to-head election against any other matching. One notable feature of popular matchings is their close relationship to \emph{stable matchings}: in one of the most basic models, a stable matching is a popular matching of minimum size. In this sense, popular matchings can be regarded as a relaxation of stable matchings that may match more agents, while preserving a global stability with respect to the preferences. 

%From the viewpoint of theoretical computer science and combinatorial optimization, questions of interest include the characterization of existence and polynomial-time computability. For stable matchings, many results on these questions have been successfully extended to matroidal generalizations \cite{fleiner2001matroid, fleiner2003fixed,  IW20, MY15, KTY18, FK16, Yokoi17}.  These extensions offer larger classes of combinatorial optimization problems, with broader range of applications, that can be solved efficiently. 

The theory of popular matchings is currently developing from the perspectives of theoretical computer science and combinatorial optimization. This paper aims to contribute to that development by investigating the limits of tractable generalizations of popular matchings and enhancing their applicability to practical problems. Specifically, we address \emph{popular matching problems with matroid constraints}, which generalize popular matchings in the same way as common independent sets of two matroids (i.e., matroid intersection) generalize bipartite matchings.  
This generalization unifies various previous models and extends the range of possible applications, such as matching problems with diversity constraints \cite{ehlers2014school} and distributional (regional) constraints \cite{kamada2015efficient}. See \cite{zhang2023reallocation} and \cite{KTY18} (and its online appendix) for lists of matroid constraints that arise naturally in real allocation and matching problems. Recent progress on the study of popularity includes polynomial-time algorithms for finding popular solutions subject to matroid constraints \cite{KMSY24, Kam17, kamiyama2020popular, csaji2022solving}. 

In some practical applications, certain aspects of the solutions can take priority over the preferences of the agents. For example, in dormitory reallocation problems, senior students who already have their rooms must be assigned new rooms which are at least as good as their previous ones (i.e., ``individual rationality'' in \cite{abdulkadirouglu1999house}). In company staff reshuffles, the primary objective is to maximize the total profit of the company, with the preferences of the workers and the departments considered secondary. 

In this paper, we address generalized models of popular matchings with matroid constraints which can represent these scenarios. We represent the above scenarios by appropriately defining the weights of the solutions. In our models, the candidate solutions are those of maximum-weight, and the objective is to find a popular one among them. 
Our main technical contribution is the development of polynomial-time algorithms in these models, which are designed by combining algorithms for unweighted models with fundamental techniques from combinatorial optimization. Note that our models are proper generalizations of unweighted models, which are special cases with all weights set to zero. Additionally, we provide some hardness results of more general problems, complementing our tractability results and clarifying the limits of tractable generalizations. 

\subsection{Previous Models}
\label{sec:previous}
In general, popular matchings are defined in bipartite graphs and have two models\footnote{Popularity is also considered in non-bipartite graphs \cite{gupta2019popular,faenza2019popular}, but that is outside the scope of this paper}. One model is the \emph{one-sided preferences model}, where only one side of the bipartite vertex set has preferences, and the other model is the \emph{two-sided preferences model}, where both sides of the vertex set have preferences. 

\paragraph*{Models on bipartite graphs.}
The one-sided preferences model of popular matchings is defined as follows. 
Let $G=(A,B;S)$ be a bipartite graph with vertex set $A\cup B$ and edge set $S$, 
where each edge in $S$ connects a vertex in $A$ and one in $B$. 
For a vertex $i \in A \cup B$, let $S (i)\subseteq S$ denote the set of the edges incident to $i$. An edge subset $M\subseteq S$ is a \emph{matching} 
if $|M \cap S(i)|\le 1$ for each $i\in A \cup B$. 
In the one-sided preferences model, each vertex $i\in A$ represents an agent who has preferences over $S(i)\cup\{\emptyset\}$, in which $\emptyset$ is the least preferred. 
This means that assigning an arbitrary element is more preferred than assigning no element. 
Let an order $\succ_i$ on $S(i) \cup \{\emptyset\}$ 
represent the preferences of an agent $i\in A$, where $\succ_i$ is either a total order, a weak order, or a partial order. (The specific type of orders will be clarified when we describe previous/our results.)
For a matching $M\subseteq S$ and a vertex $i\in A \cup B$, 
let $M(i)$ denote the unique edge $M \cap S(i)$, if it exists. 
For convenience, 
if $M \cap S(i)= \emptyset$, 
then let $M(i)$ represent $\emptyset$. 
For two matchings $M,N\subseteq S$,
define $\Delta(M,N)\in \Z$ by 
\begin{align*}
%\label{EQonesidedpm}
\Delta(M,N) = |\{\, i\in A  \colon M(i) \succ_i N(i) \,\}| - |\{\, i\in A  \colon N(i) \succ_i M(i) \,\}|.
\end{align*}
A matching $M$ is called  a \emph{popular matching} if 
$\Delta(M,N)\ge 0$ for each matching $N$ in $G$. 

The two-sided preferences model is defined in the same way. 
The difference from the one-sided preferences model is that 
each vertex $i$ in both $A$ and $B$ has preferences over $S(i)\cup \{\emptyset\}$, 
and 
the definition of $\Delta(M,N)$ above is replaced with 
\begin{align*}
%\label{EQtwosidedpm}
\Delta(M,N) = |\{\, i\in A \cup B  \colon M(i) \succ_i N(i) \,\}| - |\{\, i\in A \cup B  \colon N(i) \succ_i M(i) \,\}|.
\end{align*}

In the one-sided preferences model, not all instances admit popular matchings. Abraham et al.\ \cite{AIKM07} provided an efficient algorithm to determine the existence of a popular matching for preference lists with ties. This tractability is extended to partial order preferences in \cite{kavitha2021maximum}.

In the two-sided preferences model, every instance admits a popular matching if preferences are total orders. This is because any stable matching is inherently popular \cite{gardenfors1975match}. However, if ties are allowed, the existence of a popular matching is not guaranteed, and determining its existence is NP-hard \cite{biro2010popular}. The algorithmic research of popular matchings in the two-sided preferences model has become vibrant since Huang and Kavitha \cite{huang2011popular} proved the tractability of the maximum popular matching problem.

\paragraph*{Matroid generalizations.}
Recall that bipartite matching is a special case of matroid intersection. 
Both of the aforementioned two popular matching models have been generalized to the models with matroid constraints. 

Here we describe the generalization of the one-sided preferences model. 
For a positive integer $k$, we denote $[k]=\{1,2,\dots, k\}$. 
Let $\{S_1, S_2, \dots, S_n\}$ be a partition of a finite set $S$ and $M_1$ be a $1$-partition matroid defined by this partition. That is, $M_1=(S, \cI_1)$ is a matroid with ground set $S$ and independent set family $\cI_1\subseteq 2^S$ defined by $\cI_1=\{\, I\subseteq S: |I\cap S_i|\leq 1\ (i\in [n])\}$ (see Section \ref{sec:matroids} for the definition of matroids).

Each index $i\in [n]$ represents an agent, 
and has an order $\succ_i$ on $S_i\cup \{\emptyset\}$ satisfying $u\succ_i \emptyset$ for each element $u\in S_i$. 
Additionally, we have another matroid $M_2=(S,\cI_2)$, which can be an arbitrary matroid and has no associated orders. 
A set $I\in \cI_1 \cap \cI_2$ is referred to as a \emph{common independent set} of $M_1$ and $M_2$. 

The popularity of common independent sets is defined similarly to that of popular matchings. 
For a common independent set $I\in \cI_1\cap \cI_2$ and an agent $i\in[n]$, let $I(i)$ denote the unique element in $I\cap S_i$ if it exists, and $\emptyset$ otherwise. Given any pair of common independent sets $I, J\in \cI_1\cap \cI_2$, define 
$\Delta(I,J) \in \Z$ by 
%we say that an agent $i$ prefers $I$ to $J$ if $I(i)\succ_i J(i)$, where any element $u\in S_i$ is regarded to satisfy $u\succ_i \emptyset$. Let 
%%\[\Delta(I,J)=|\{\, i\in [k]: I(i)\succ_i J(i)\,\}|-|\{\, i\in [k]: J(i)\succ_i I(i)\,\}|.\] 
\begin{align*}
\Delta(I,J)=|\{\, i\in [n]: I(i)\succ_i J(i)\,\}|-|\{\, i\in [n]: J(i)\succ_i I(i)\,\}|.
\end{align*} 
A common independent set $I\in \cI_1\cap \cI_2$ is {\em popular} if $\Delta(I, J)\geq 0$ holds for every common independent set $J\in \cI_1\cap\cI_2$. 

It was shown by Kavitha et al.~\cite{KMSY24} that one can determine the existence of a popular common independent set even for partial order preferences. This result is a common generalization of various previously known tractability results on popular matchings \cite{AIKM07}, popular branchings \cite{KKMSS20}, and on popular matchings with matroid constraints \cite{Kam17}. 

In the two-sided preferences model with total orders, the structural and tractability results have been extended to a general model with matroid constraints. Kamiyama \cite{kamiyama2020popular} introduced the concept of popularity on matroid intersection and proved that a stable common independent set (i.e., \emph{matroid kernel} \cite{fleiner2001matroid,fleiner2003fixed}, defined in Section~\ref{sec:prelimi2}) is popular. Since the definition of popularity in general matroid intersection is not so trivial, we defer it to Section~\ref{sec:prelimi2}.  
Intuitively, it represents popularity in a many-to-many matching model 
where each voter has a matroid constraint and casts multiple votes while respecting that constraint. 
%To this setting, the tractability of the maximum popular matching problem extends \cite{kamiyama2020popular,csaji2022solving}. 
%These previous works on the two-sided preferences model consider total order preferences; if ties are allowed, computation of a popular matching becomes NP-hard even in the simple bipartite matching case~\cite{BB20}.

\subsection{Our Contributions}
\label{sec:contribution}

For both the one- and two-sided preferences models, 
we introduce \emph{weights} to common independent sets  
and 
address the problem of finding a maximum-weight common independent set 
that is popular among all common independent sets with maximum weight. 
Throughout the paper, we assume that independence oracles of matroids are available.

Our problem for the one-sided preferences model is described as follows. 
Again, 
we are given a 1-partition matroid $M_1=(S,\cI_1)$ defined by a partition $\{S_1,S_2,\ldots, S_n\}$ of $S$, 
another matroid $M_2=(S,\cI_2)$, 
and 
a partial order $\succ_i$ on $S_i \cup \{\emptyset\}$ for each $i\in [n]$, where $\emptyset$ is the least preferred.
In this model, we assume that 
the orders are \emph{partial orders}. 
A \emph{partial order} is defined as an irreflexive, asymmetric, and transitive binary relation. 
In addition, now a weight function $w\colon S \to \R$ is defined on the ground set $S$. 
For a common independent set $I\subseteq S$, 
its weight $w(I)$ is defined as 
$w(I)= \sum_{u\in I}w(u)$. 
Let $\opt(w)$ denote the maximum weight of a common independent set, i.e., $\opt(w)=\max\{\,w(I): I\in \cI_1\cap \cI_2\,\}$. 
%The objective of the problem is to find a maximum-weight common independent set 
% which is popular among all of the maximum-weight common independent sets. 
%The {\em popular maximum-weight common independent set problem in the one-sided preferences model} (\popoptone, for short) asks to find a popular solution among the set of maximum-weight common independent sets.

\begin{definition}%[The popular maximum-weight common independent set problem in the one-sided preferences model]
A common independent set $I \in \cI_1\cap \cI_2$ is called a {\em popular maximum-weight common independent set} if $w(I) = \opt(w)$ and $\Delta(I,J)\ge 0$ for any $J\in \cI_1\cap \cI_2$ with $w(J)=\opt(w)$.
%, or report the non-existence of such $I$.
\end{definition}

By appropriately setting the weight function, popular maximum-weight common independent sets can describe previously investigated solution concepts, such as popular common independent sets \cite{Kam17} ($w(u)=0$ for all $u\in S$) and popular common bases~\cite{KMSY24} ($w(u)=1$ for all $u\in S$). This fact implies that there are instances that admit no popular maximum-weight common independent sets. The {\em popular maximum-weight common independent set problem in the one-sided preferences model} asks to determine the existence of a solution and to find one if it exists.
Our first technical contribution is a polynomial-time algorithm to solve this problem with general weight functions. 

\begin{theorem}[Tractability in the weighted model with one-sided preferences]
\label{thm:weighted}
%\popoptone\ is tractable. That is, 
Given a $1$-partition matroid $M_1=(S, \cI_1)$ associated with partial orders $\{\succ_i\}_{i\in [n]}$, an arbitrary matroid $M_2=(S, \cI_2)$, and a weight function $w\colon S \to \R$, one can determine the existence of a popular maximum-weight common independent set and find one if it exists in polynomial time.
\end{theorem}
We remark that 
%%it is impossible 
there is little hope 
to extend this result to the model where not only $M_2$ but also $M_1$ is an arbitrary matroid, because it is NP-hard to determine the existence of a popular $b$-matching (i.e.,\ intersection of two arbitrary partition matroids) \cite{paluch2014popular, csaji2024popularity}.

We then extend the weight functions to a broader class of utility functions. 
We address \emph{M$^\natural$-concave utility functions}, 
a primary class of discrete concave functions \cite{Mbook}. See Section~\ref{sec:matroids} for a precise definition.
Since it was shown in~\cite{FY03} that M$^\natural$-concavity is equivalent to the Kelso--Crawford {\em gross substitute condition}~\cite{kelso1982job}, 
M$^\natural$-concave functions have been studied in the context of two-sided markets \cite{fujishige2007two,murota2015lattice}. 
Two-sided markets in which one side has preferences and the other has an M$^\natural$-concave function are studied in \cite{KTY18}.

The \emph{popular maximum-utility common independent set problem} is formulated by replacing the weight function $w\colon S\to \R$ in the popular maximum-weight common independent set problem with an M$^\natural$-concave function $f \colon \I_2\to \R$. 
Let $\opt(f)$ denote the maximum utility of a common independent set, i.e., $\opt(f)=\max\{\,f(I): I\in \cI_1\cap \cI_2\,\}$.  

\begin{definition}%[The popular maximum-utility common independent set problem in the one-sided preferences model]
A common independent set $I \in \cI_1\cap \cI_2$ is called a {\em popular maximum-utility common independent set} if $f(I) = \opt(f)$ and 
$\Delta(I,J)\ge 0$ for any $J\in \cI_1\cap \cI_2$ with $f(J)=\opt(f)$.
\end{definition}
Note that this is a proper generalization of a maximum-weight common independent set, because a modular function on the independent set family of a matroid is M$^\natural$-concave. We devise a polynomial-time algorithm for this generalized problem utilizing structural results known for M$^\natural$-convex functions. (Though Theorem~\ref{thm:weighted} follows from this generalized result, we provide a separate proof for the weighted case as it relies on more basic tools and may be more accessible.)

\begin{theorem}[Tractability in the M$^\natural$-concave model with one-sided preferences]
\label{thm:M-conc}
%\popoptoneM\ is tractable.That is, 
Given a $1$-partition matroid $M_1=(S, \cI_1)$ associated with partial orders $\{\succ_i\}_{i\in [n]}$, an arbitrary matroid $M_2=(S, \cI_2)$, and an M$^\natural$-concave function $f\colon \cI_2 \to \R$, one can determine the existence of a popular maximum-utility common independent set and find one if it exists in polynomial time.
\end{theorem}

In 
proving Theorems~\ref{thm:weighted} and \ref{thm:M-conc}, 
we reduce the problems 
%%to the popular common base problem, considered and solved in Kavitha et al.~\cite{KMSY24}. 
to the popular common base problem \cite{KMSY24} 
by defining new matroids whose common bases correspond to the maximum-weight common independent sets of the original matroids. 
The reduction for Theorem~\ref{thm:weighted} employs
LP duality and complementary slackness 
for the matroid intersection polytope, 
while that for Theorem~\ref{thm:M-conc} relies
on the weight splitting theorem 
for M$^\natural$-convex intersection problem. 
%While these technical ingredients look different, they have similar substances, and we have adopted a proof of Theorem~\ref{thm:weighted} not relying on weight splitting in order to make it more accessible.

\bigskip
For the two-sided preferences model, we address the problem of finding a popular maximum-weight matching in a many-to-many matching setting with two-sided preferences and matroid constraints. 
In this model, two matroids are given on the same ground set $S$, both as direct sums: $M_1= M^1_1 \oplus M^1_2 \oplus \dots \oplus M^1_{k_1}$ and $M_2= M^2_1 \oplus M^2_2 \oplus \dots \oplus M^2_{k_2}$. Each summand $M^i_j=(S^i_j, \cI^i_j)$ corresponds to an agent, and hence there are $k_1+k_2$ agents. A set $I\subseteq S$ is feasible if $I\cap S^i_j\in \cI^i_j$ for each $i\in \{1,2\}$ and $j\in [k_i]$. The simple bipartite matching model is a special case where each $M^i_j=(S^i_j, \cI^i_j)$ is a uniform matroid of rank $1$.

A detailed description of the problem is provided in Section \ref{sec:prelimi2}. It is worth mentioning that our {\em popular maximum-weight common independent set problem in the two-sided preferences model} encompasses  previously studied problems in the two-sided preferences model, such as the popular common independent set problem \cite{kamiyama2020popular, csaji2022solving}, the popular common base problem \cite{Kuffner23}, and the popular critical matching problem \cite{kavitha2021matchings}.

\begin{theorem}[Tractability in the weighted model with two-sided preferences]
\label{thm:twosided}
In the two-sided preferences model, if preferences are total orders, then a popular maximum-weight common independent set always exists and one can find it in polynomial time.
%there is a polynomial-time algorithm to find one.
\end{theorem}
This theorem assumes that preferences are total orders. They probably cannot be extended to more general orders because, when ties are allowed, finding a popular matching is NP-hard, even in the simple bipartite matching model \cite{biro2010popular}.

Similarly to the one-sided preferences model, our algorithm for this problem relies on a characterization of the maximum-weight common independent sets obtained from complementary slackness conditions. However, this case is more challenging, 
and the same reduction cannot be applied.
In the one-sided preferences model, matroids are used only to define the set of the feasible matchings, i.e., ``candidates,'' 
so popularity remains unchanged in the reduced instance of the popular common base problem. 
%%we can reduce our problem to the popular common base problem by defining new matroids whose common bases correspond to the maximum-weight common independent sets of the original matroids. 
In contrast, in the two-sided preferences model, where each voter casts multiple votes respecting her constraint, matroids are also used to define popularity, i.e., the ``election,''  and hence the popularity may change in the reduction. 
%%Therefore, we can not apply the same reduction as the one-sided case.

%%To solve our weighted problem in the two-sided preferences model, 
In order to resolve this issue, 
we introduce a new problem, which we call the {\em popular critical common independent set} problem. We show that the popular maximum-weight common independent set problem can be reduced to this new problem and provide an efficient algorithm to solve it. 

\medskip
We also investigate further generalizations of our problems and provide hardness results.
In the problems solved in Theorems~\ref{thm:weighted}, \ref{thm:M-conc}, and \ref{thm:twosided}, our objective is to find a common independent set that is popular within the set of ``optimal'' common independent sets.
Natural variants of them are problems to find a common independent set that is popular within the set of ``near-optimal'' common independent sets. Our hardness results hold even for the simple bipartite matching case. Suppose that we are given a bipartite graph $G=(A,B;E)$ and a weight function $w :E\to \R$. 
In the one-sided preferences model (resp. two-sided preferences model), we are given partial orders $\{\succ_i\}_{i\in A}$ (resp., total orders $\{\succ_i\}_{i\in A\cup B}$). In addition, we are given $k\in \R$, a threshold.

\begin{definition}%[The popular near-maximum-weight matching problem in the one-sided (resp., two-sided) preferences model]
A matching $M\subseteq E$ is called a {\em popular near-maximum-weight matching} if $w (M)\ge k$ and $\Delta (M,N)\ge 0$ for any matching $N$ with $w (N)\ge k$.
\end{definition}
As we will show in Section~\ref{sec:hardness}, the existence of a popular near-maximum-weight matching is not guaranteed even in the two-sided preferences model.
The {\em popular near-maximum-weight matching problem in the one-sided (resp., two-sided) preferences model} asks to determine the existence of a popular near-maximum-weight matching and to find one if it exists.

For the one-sided preferences model, we demonstrate that this problem is as hard as the {\em exact matching} problem \cite{papadimitriou1982complexity} (see Section~\ref{sec:hardness1} for the definition), for which the existence of a deterministic efficient algorithm remains a longstanding open question.

\begin{theorem}\label{thm:hardness-1-sided}
If there exists a deterministic polynomial-time algorithm for the popular near-maximum-weight matching problem in the one-sided preferences model, then there exists a deterministic polynomial-time algorithm for the exact matching problem. This holds even if preferences are weak orders (i.e., lists with ties) and weights are limited to values in $\{0,1\}$. 
\end{theorem}

%For the two-sided preferences model, we provide the following NP-hardness result.
Note that if weights are all $1$ (i.e., $w(e)=1$ for every edge $e\in E$), then the popular near-maximum-weight matching problem asks to find a matching that is popular among matchings of size at least $k$. For the one-sided preferences model, this special case can be solved via a reduction to the popular assignment problem as shown in \cite[Section 2.3]{KKMSS22}. In contrast, for the two-sided preferences model, even the cardinality constrained version is NP-hard. 
%\begin{theorem}\label{thm:hardness-2-sided}
%The popular near-maximum-weight matching problem in the two-sided preferences model is NP-hard. This holds even if weights are limited to values in $\{0,1\}$. 
%\end{theorem}
\begin{theorem}\label{thm:hardness-2-sided-allone}
The popular near-maximum-weight matching problem in the two-sided preferences model is NP-hard even if weights are all $1$. That is, it is NP-hard to determine the existence of a matching $M$ such that $|M|\ge k$ and $\Delta (M,N)\ge 0$ for any matching $N$ with $|N|\ge k$. 
\end{theorem}
We remark that it was shown by Kavitha~\cite[Theorem 3]{kavitha2014size} that, for any integer $\ell\geq 2$, one can find a matching $M_{\ell}$ such that $|M_{\ell}|\geq \frac{\ell}{\ell+1}|M_{\rm max}|$ (where $M_{\rm max}$ is a maximum matching) and $\Delta (M_{\ell},N)\ge 0$ for any matching $N$ with $|N|\ge |M_\ell|$. Such a matching $M_\ell$ can be seen as a popular near-maximum-weight matching with all $1$ weights and $k=|M_\ell|$. Our result does not contradict this fact because in Theorem~\ref{thm:hardness-2-sided-allone} the threshold $k$ can be chosen arbitrarily.

As these two theorems show, the popular near-optimal matching problems are difficult, contrasting with the fact that popular optimal solutions can be efficiently computed even in general matroidal settings with general weights (Theorems~\ref{thm:weighted}, \ref{thm:M-conc}, and \ref{thm:twosided}).

%One explanation for this gap could be the difference in structural properties between optimal and near-optimal solutions of a matroid. While the set of maximizers of an M$^\natural$-concave function (which includes weighted matroids) forms a matroid-like structure (see Lemma~\ref{LEMargmax} in Section~\ref{sec:matroids}), it is known that the upper level set of an M$^\natural$-concave function does not have that property~\cite{shioura2000level}.

%\textcolor{blue}{These results present a clear contrast to the fact that there is a polynomial-time algorithm to find a matching $M$ such that $|M|\ge k$ and $\Delta(M,N)\ge 0$ for any matching $N$ with $|N|\geq k$ in each of the one-sided  and two-sided preferences models \cite{KMSY24, kavitha2014size}. Thus, there is a gap between the problems of finding a popular solution under cardinality lower bound constraint and those under weight lower bound constraint, even for $\{0,1\}$-weights.}
\subsection{Related Works}
\label{SECrelated}
The concept of popularity in the context of matchings with two-sided preferences was introduced by G\"{a}rdenfors \cite{gardenfors1975match}, who showed that every stable matching is popular. Popularity in the one-sided model was considered only much later by Abraham et al.\ \cite{AIKM07}, 
where the authors gave an efficient algorithm to find a popular matching.
%where the authors gave a characterization of popular matchings and an $O(\sqrt{n}m)$ time algorithm for the maximum cardinality popular matching problem, where $n$ and $m$ are the numbers of vertices and edges, respectively. 
%
After these results, research on the topic gained momentum and led to several generalizations and new approaches in both the one- and two-sided cases. Here we list the ones most relevant to our present work.

\paragraph{One-sided Preferences.} A natural generalization of bipartite matching is the bipartite $b$-matching problem. Manlove and Sng \cite{manlove2006popular} showed that the problem is tractable if only the side without preferences has capacities. 
However, when the capacities are on the side with preferences, then determining the existence of a popular matching becomes NP-complete, as shown by Paluch \cite{paluch2014popular} for preference lists with ties and by Csáji \cite{csaji2024popularity} for strictly ordered lists.

Continuing the line of research of \cite{manlove2006popular}, Kamiyama \cite{Kam17} considered the generalization where the feasible choices for each vertex on the side without preferences are determined by matroid constraints.
%, and the preference lists on the other side have ties.
A related, but distinct problem is the popular branching problem \cite{KKMSS20}, where the feasible solutions are branchings of a directed graph, and vertices have preferences on the incoming arcs. 
Another direction of research relevant to our paper is the restriction of feasible solutions based on cardinality. 
In the popular assignment problem introduced in \cite{KKMSS22}, only perfect matchings are considered to be feasible, so the objective is to find one that is popular among the perfect matchings.
A common generalization of the above problems was solved in \cite{KMSY24}, where the popular common base problem was shown to be tractable for arbitrary partial order preferences; we will state this result as Theorem~\ref{thm:pop-base} in Section~\ref{sec:onesided}.

\paragraph{Two-sided Preferences.} In contrast to the one-sided case, the problem with two-sided preferences becomes NP-hard if ties are allowed in the preference lists (even for ties on one side), as shown in \cite{biro2010popular,cseh2017popular}. 
%Therefore, we only consider preferences without ties. 
Maximum-size popular matchings, and their subclass called dominant matchings, have been analyzed in several papers \cite{huang2011popular,kavitha2014size,cseh2018popular}. Kavitha \cite{kavitha2014size} showed how to find a maximum matching that is popular among the maximum matchings; furthermore, she considered the class of critical matchings \cite{kavitha2021maximum}, which contain as many vertices from a given set as possible, and solved the problem of finding a maximum-size matching among those that are popular among the critical matchings. Very recently, the popular maximum-weight matching problem on bipartite graphs was studied by Kavitha\cite{Kavitha2024optimal}.

Concerning popularity in the many-to-many setting, the definition of voting is less obvious; a model and efficient algorithms have been developed by Brandl and Kavitha\cite{brandl2018popular,brandl2019two}. The model was extended by using matroid constraints by Kamiyama \cite{kamiyama2020popular}. The maximum-size popular matching problem in this model was solved in \cite{csaji2022solving}.
Matroidal generalizations of stable matchings date back to the work of Fleiner \cite{fleiner2001matroid, fleiner2003fixed}, who defined the matroid kernel problem as a natural generalization of bipartite stable matchings, and showed that an elegant generalization of the Gale-Shapley algorithm efficiently finds a matroid kernel. Algorithms for popular matchings that involve matroid constraints usually rely on some version of Fleiner's algorithm.

\subsection*{Paper Organization}
The rest of the paper is organized as follows.
Section~\ref{sec:matroids} describes some basics on matroids and M$^\natural$-concave functions. Section~\ref{sec:onesided} is devoted to the proofs of Theorems~\ref{thm:weighted} and \ref{thm:M-conc}, the tractability results on the one-sided preferences model. In Section~\ref{sec:twosided}, we precisely define our two-sided preferences model and show Theorem~\ref{thm:twosided}.
Theorems~\ref{thm:hardness-1-sided} and \ref{thm:hardness-2-sided-allone}, the hardness results on popular near-optimal matchings are shown in Section~\ref{sec:hardness}.

\section{Matroids and M$^\natural$-concave Functions}\label{sec:matroids}
For a set $X$ and an element $x$, 
we use the notations $X-x=X\setminus \{x\}$ and $X+x=X\cup \{x\}$. 

A pair $(S,\cI)$ 
of a finite set $S$
and a nonempty family $\cI \subseteq 2^S$ is called a \emph{matroid} if it satisfies the following axioms: 

\begin{enumerate}
    \item[\rm (I1)] If $I \in \cI$  and $I' \subseteq I$, then $I'\in \cI$, 
    \item[\rm (I2)] If $I,I'\in \cI$ and $|I'| < |I|$, then $I'+x \in \cI$ for some $x\in I \setminus I'$. 
\end{enumerate}

The set $S$ is referred to as the \emph{ground set} 
and $\cI$ as the \emph{independent set family}. 
A set in $\cI$ is referred to as an \emph{independent set}. 

The \emph{rank function} $r \colon 2^S \to \Z$ of a matroid $(S,\cI)$ is defined as 
\begin{align*}
r(X) = \max \{\,|Z| \colon Z \subseteq X, Z \in \cI\,\} \quad (X \subseteq S). 
\end{align*}

A \emph{base} $B \subseteq S$ of a matroid $(S,\cI)$ is an inclusionwise maximum independent set. 
The \emph{base family} of a matroid $(S,\cI)$, 
i.e.,\ 
the set of all bases in $(S,\cI)$ is often denoted by $\B$. 
The base family $\B\subseteq 2^S$ uniquely determines the original matroid $(S,\cI)$, 
and thus a matroid is often denoted by a pair $(S,\B)$ of its ground set and base family. 
Observe that 
$|B_1| = |B_2|$ holds for any bases $B_1,B_2\in \B$ by the axiom (I2). 

For a finite set $S$ and a nonempty family $\B \subseteq 2^S$, 
a pair $(S, \B)$ is a matroid with ground set $S$ and base family $\B$ if and only if the following axiom is satisfied:
\begin{enumerate}
    \item[\rm (B1)] $B,B' \in \B$  and $x \in B \setminus B'$ implies that 
    $B - x + y \in \B$ for some $y\in B' \setminus B$.
\end{enumerate}

Here we describe some basic operations on matroids, which will be used in our proofs.
Let $M=(S, \cI)$ be a matroid, 
and let $T\subseteq S$ be a subset of $S$. 
\begin{itemize}
\item Define a family $\cI' \subseteq 2^T$ by 
$\cI'=\{\, X\subseteq T \colon X\in \cI\,\}$. 
Then, $(T, \cI')$ is a matroid called the {\em restriction} of $M$ to $T$. 
\item Let $B_T$ be any maximal subset of $T$ in $\cI$ (i.e., a base of $T$), 
and define a family $\cI'' \subseteq 2^{S\setminus T}$ by $\cI''=\{\, X\subseteq S\setminus T \colon B_T\cup X\in \cI\,\}$. Then, $(S\setminus T, \cI'')$ is a matroid, and this operation is called {\em contracting} $T$. 
Note that the family $\cI''$ is not affected by the choice of $B_T$. 
\item Let $k$ be a positive integer and define $\cI^k\subseteq 2^S$ by $\cI^k=\{\,X: X\in \cI, |X|\leq k\,\}$. Then, $(S, \cI^k)$ is a matroid, which we call the {\em $k$-truncation} of $M$. 
\item For matroids $M_1=(S_1, \cI_1),M_2=(S_2, \cI_2),\dots, M_k=(S_k, \cI_k)$ such that $S_i~(i\in[k])$ are mutually disjoint, let $S^{\star}\coloneqq S_1\cup S_2\cup \cdots \cup S_k$ and $\cI^{\star}=\{\, X\subseteq S^{\star}:X\cap S_i\in \cI_i\ (i\in[k])\,\}$. Then, $(S^{\star}, \cI^{\star})$ is a matroid called the {\em direct sum} of $M_i~(i\in [k])$ and denoted by $M_1\oplus M_2\oplus \cdots \oplus M_k$.
\end{itemize}

In the literature, the concept of matroids is generalized to that of \emph{generalized matroids} \cite{Tardos85}, which are known to be equivalent to \emph{M$^\natural$-concave families}.
Hereafter, our discussion involving M$^\natural$-concave utility functions are described in terms of \emph{M$^\natural$-concave families}. 

% For a set $X$ and an element $x$, 
% we use the notations $X-x=X\setminus \{x\}$ and $X+x=X\cup \{x\}$. 
Let $\J$ be a nonempty family of subsets of a finite set $S$.
We say that $\J$ is an {\em M$^\natural$-convex family} if, for any $X, Y\in \J$ and $x\in X\setminus Y$, at least one of the following holds: 
\begin{enumerate}
    \item[\rm (i)] $X-x\in \J$, $Y+x\in \J$. 
    \item[\rm (ii)] There exists some element $y\in Y\setminus X$ such that $X-x+y\in \J$, $Y+x-y\in \J$. 
\end{enumerate}

%It holds that a family $\J\subseteq 2^S$ is M$^\natural$-convex if and only if $(S, \J)$ is a generalized matroid. 
Unlike the independent set family of a matroid, an M$^\natural$-convex family is not required to have the hereditary property (I1) while
it is known to satisfy the augmentation property (I2) (see \cite[Lemma 2.4]{Tardos85}, \cite[Theorem 1.1]{murota2018simpler}). From this, it follows that the independent set family of a matroid is characterized as an M$^\natural$-convex family containing $\emptyset$.

We can also observe that an M$^\natural$-convex family gives rise to the base family of a matroid as follows. Although this is a known fact  \cite[Theorem 2.9]{Tardos85}, we provide a proof for completeness.
\begin{lemma}\label{lem:projection}
Let $\J\subseteq 2^S$ be an M$^\natural$-convex family, $D$ a finite set disjoint from $S$, and $t$ any positive integer. If $\B=\{\,B\subseteq S\cup D: B\cap S\in \J,\, |B|=t\,\}$ is nonempty, then it forms the base family of a matroid on the ground set $S\cup D$.
\end{lemma}
\begin{proof}
Assume that $\B\neq \emptyset$. 
We prove that 
$\B$ satisfies (B1) by 
using the fact that an M$^\natural$-convex family satisfies the augmentation axiom (I2). 
%, we show that for any $B_1, B_2\in \B$ and $x\in B_1\setminus B_2$, there exists  $y\in B_2\setminus B_1$ such that $B_1-x+y\in \B$. 
Let $B,B'\in \B$ and $x\in B\setminus B'$. 
Define $X=B\cap S$ and $Y=B'\cap S$. 
Clearly, $X, Y\in \J$.

We first consider the case $x\in D$. 
If $(B'\setminus B)\cap D \neq \emptyset$, then any $y\in (B'\setminus B)\cap D$ satisfies $B-x+y\in \B$. 
If $(B'\setminus B)\cap D = \emptyset$, it follows from $x\in B\cap D$ and $|B|=|B'|$ that $|X|<|Y|$, and hence (I2) implies that there is an element $y\in Y\setminus X$ such that $X+y\in \J$. 
This element $y$ satisfies that $y\in B'\setminus B$ and $B-x+y\in \B$.

We next consider the case $x\in B\setminus D$, i.e., $x\in X$.
By the definition of an M$^\natural$-convex family, we have that $X-x\in \J$ or $X-x+y\in \J$ for some $y\in Y\setminus X\subseteq B'\setminus B$. In the latter case, we clearly have $B-x+y\in \B$. We then assume $X-x\in \J$. If $(B'\setminus B)\cap D\neq \emptyset$, any element $y\in (B'\setminus B)\cap D$ satisfies $B-x+y\in \B$. If $(B'\setminus B)\cap D= \emptyset$, then we have $|X|\leq |Y|$, and hence $|X-x|<|Y|$. 
%%By (I2), 
Since $\J$ satisfies (I2) and $X-x,Y \in \J$,
there exists an element $y\in Y\setminus (X-x)$ such that $X-x+y\in \J$. 
This element $y$ satisfies that $y\in B'\setminus B$ and $B-x+y\in \B$.
\end{proof}

M$^\natural$-concave functions are defined as a quantitative generalization of M$^\natural$-convex families. 
While there are various equivalent definitions of M$^\natural$-concavity, 
we adopt the following definition due to Murota~\cite{murota2016discrete}. 
Let $S$ be a finite set and $\J\subseteq 2^S$ be a family of subset of $S$. 
We say that a function $f:\J\to \R$ is an {\em M$^\natural$-concave function} if at least one of the following holds for any subsets $X, Y\subseteq \J$ and any element $x\in X\setminus Y$:
\begin{description}
\item[\rm (i)] $X-x\in \J$, $Y+x\in \J$ and $f(X)+f(Y)\leq f(X-x)+f(Y+x)$.
\item[\rm (ii)] There exists some element $y\in Y\setminus X$ such that $X-x+y\in \J$, $Y+x-y\in \J$, and $f(X)+f(Y)\leq f(X-x+y)+f(Y+x-y)$.
\end{description}

It follows from this definition that the domain $\J$ of an M$^\natural$-concave function $f$ must be an M$^\natural$-convex family.
The following properties of M$^\natural$-concave functions will be useful in our proof.
For a function $f:\J\to \R$, where $\J \subseteq 2^S$, and a vector $q\in \R^S$, define a function $f[q]:\J\to \R$ by $f[q](X)=f(X)+\sum_{u\in X}q(u)$ ($X\in \J$). 
\begin{lemma}[e.g., Murota \cite{Mbook}]
\label{LEMargmax}
Let $f:\J\to \R$ be an M$^\natural$-concave function, 
where $\J \subseteq 2^S$. 
\begin{itemize}
\item 
For any vector $q\in \R^S$, a function $f[q]:\J\to \R$ 
is M$^\natural$-concave. 
\item 
The set of maximizers of $f$, i.e., $\arg\max f \subseteq 2^S$, forms an M$^\natural$-convex family.
\end{itemize}
\end{lemma}
It is known that an M$^\natural$-concave function can be maximized efficiently (assuming that a value oracle is available). 
While the sum of two M$^\natural$-concave functions is not necessarily M$^\natural$-concave, it is also known to be 
maximized efficiently, which generalizes the fact that weighted matroid intersection is tractable. 
The set of maximizers of the sum of two M$^\natural$-concave functions is characterized by the following structure theorem\footnote{The original theorem by Murota \cite[Theorem 4.1]{Murota96_1} shows a stronger result for valuated matroid intersection. A version described in terms of M$^\natural$-concave functions can be found in Murota~\cite{murota2016discrete}. The statement in Lemma~\ref{LEMsplit} is obtained by applying Theorem 11.2(2) in \cite{murota2016discrete} with $w$ being a constantly zero function.}.
\begin{lemma}[Murota~\cite{Murota96_1,Murota96_2}]
\label{LEMsplit}
For two M$^\natural$-concave functions $f_1:\J_1\to \R$ and $f_2: \J_2\to \R$, where $\J_1, \J_2\subseteq 2^S$, there exists a vector $p\in \R^S$ such that 
\[\arg \max(f_1+f_2)=\arg\max(f_1[+p])\cap \arg\max(f_2[-p]).\]
Furthermore, such a vector $p$ can be computed efficiently assuming that membership oracles of $\J_1, \J_2$ and value oracles of $f_1, f_2$ are available and some members of $\J_1$ and $\J_2$ are known.
\end{lemma}

\section{One-sided Preferences Models}
\label{sec:onesided}
%In this section, we prove our results on the one-sided preferences model described in Section~\ref{sec:previous}.
We show Theorems~\ref{thm:weighted} and \ref{thm:M-conc} in Sections~\ref{sec:additive} and \ref{sec:Mconcave}, respectively.
We remark that 
our proof of Theorem~\ref{thm:M-conc}
is not a direct extension of the proof of Theorem~\ref{thm:weighted},
although the weight $w(I) = \sum_{u\in I}w(u)$ is a special case of an M$^\natural$-concave function. 
In both of the proofs, we reduce our problems to the popular common base problem, described below.

As in Section~\ref{sec:previous}, we let $M_1=(S, \cI_1)$ be a 1-partition matroid 
defined by a partition $\{S_1,S_2,\ldots, S_n\}$ of $S$, 
associated with a partial order $\succ_i$ on $S_i \cup \emptyset$ for each $i\in [n]$, 
and let $M_2=(S,\cI_2)$ be an arbitrary matroid. 
Denote the base family of $M_1$ by $\B_1$, and 
that of $M_2$ by $\B_2$. 
Assume that $\B_1\cap \B_2 \neq \emptyset$. 
A member $I\in \B_1\cap \B_2$ is called a {\em popular common base} if $\Delta(I,J)\geq 0$ for every common base $J\in \B_1\cap \B_2$, where $\Delta(I,J)$ is defined by
$$\Delta(I,J)=|\{\, i\in [n]: I(i)\succ_i J(i)\,\}|-|\{\, i\in [n]: J(i)\succ_i I(i)\,\}|$$ as in Section~\ref{sec:previous}. The {\em popular common base problem} asks to determine the existence of a popular common base and to find one if it exists. A polynomial-time algorithm to solve this problem was recently proposed. %The following theorem has been shown recently.

\begin{theorem}[Kavitha--Makino--Schlotter--Yokoi\cite{KMSY24}]\label{thm:pop-base}
Given a $1$-partition matroid $M_1=(S, \cI_1)$ associated with partial orders $\{\succ_i\}_{i\in [n]}$ and an arbitrary matroid $M_2=(S, \cI_2)$, one can determine the existence of a popular common base and find one if it exists in polynomial time. 
%The algorithm works for partial order preferences.
\end{theorem}

\begin{comment}
\begin{remark}
Note that what we want to find in Theorem~\ref{thm:weighted} is a popular solution in the family of maximum-weight common independent sets (i.e., $\{\, I\in \cI_1\cap \cI_2 : w(I)\geq w(J)~(J\in \cI_1\cap \cI_2)\,\}$), not the intersection of the families $\{\, I\in \cI_j : w(I)\geq w(J)~(J\in \cI_j)\,\}\,(j=1,2)$ of maximum weight independent sets. %(each of which forms an M$^\natural$-convex family and hence can be transformed to the base family of a matroid by Lemma~\ref{lem:projection}). 
Therefore, though it is known that the family of maximum weight independent sets forms M$^\natural$-convex family, this fact (with respect to the original weights) cannot be used to solve our problem. A reduction that deals with the two matroids completely separately will not work for our problem.
\todo{Is this remark necessary? This is added to respond to reviewer's comment.}
\end{remark}
\end{comment}

%%\subsection{Additive Utility Functions}
\subsection{Finding a Popular Maximum-Weight Common Independent Set}\label{sec:additive}
% In this section, we consider the case where we have a weight function $w:S\to \R$ on the elements and the utility of a subset $I\subseteq S$ is given by $w(I)=\sum_{u\in I}w(u)$ for any $I\subseteq S$. 
% Let $\opt(w)$ be the maximum weight of a common independent set, i.e., $\opt(w)=\max\{\,w(I): I\in \cI_1\cap \cI_2\,\}$. We call a common independent set $I\in \cI_1\cap \cI_2$ a {\em popular maximum-weight-common-independent-set} if $w(I)=\opt(w)$ and $\Delta(I,J)\geq 0$ for every $J\in \cI_1\cap \cI_2$ with $w(J)=\opt(w)$. 

% Note that the problem of finding a popular maximum-weight-common-independent-set generalizes both the popular common independent set problem (where $w(u)=0$ for every $u\in S$) and the popular common base problem (where $w(u)=1$ for every $u\in S$). We show that the problem is tractable for general weight functions. 

In reducing the popular maximum-weight common independent set problem to the popular common base problem, 
we will make use of a dual optimal solution of weighted matroid intersection such that 
its support is a \emph{chain}. 
A family $\C \subseteq 2^S$ of subsets of $S$ is referred to as 
a \emph{chain} if, for any distinct $C, C'\in \C$, 
it holds that $C\subsetneq C'$ or $C'\subsetneq C$. 

\begin{proof}[Proof of Theorem \ref{thm:weighted}]
Recall that our input consists of a $1$-partition matroid $M_1=(S, \cI_1)$ associated with partial orders $\{\succ_i\}_{i\in [n]}$, another matroid $M_2=(S, \cI_2)$, and a weight function $w\colon S \to \R$.

Consider the linear programming problem \ref{LP1} 
with variables $\vec{x} \in \R^S$ 
described below, 
in which $r:2^S\to \Zp$ is the rank function of the matroid $M_2=(S, \I_2)$. 
\ref{LP1} represents weight maximization over the matroid intersection polytope of $\cI_1\cap \cI_2$, 
and 
\ref{LP2} is the dual of \ref{LP1}, 
with variables $\vec{y}\in \R^{2^S}$ and $\vec{\alpha} \in \R^{[n]}$.
\begin{table}[h]
\hspace{-5mm}
\begin{minipage}[t]{0.42\linewidth}\centering
\begin{align}
\displaystyle \mbox{Max.}& \quad \sum_{u \in S} w(u)\cdot x_u & \tag{LP1} \label{LP1}\\ 
\notag \text{s.t.}& \quad \sum_{u\in S_i}x_u \leq 1 &  (i \in [n]),\\ \notag
&\quad \sum_{u\in X} x_u  \leq r(X) & (X \subseteq S),\\ 
\notag
&\quad \qquad x_u \geq 0 & (u \in S).
\end{align}
\end{minipage}
%\hspace{0.4cm}
\hspace{0.05\textwidth}
\begin{minipage}[t]{0.54\linewidth}\centering
\begin{align}
\omit\rlap{$ \displaystyle \mbox{Min.} ~~\sum_{X \subseteq S }r(X)\cdot y_X + \sum_{i\in [n]}\alpha_i$} & & \tag{LP2} \label{LP2} \\ 
\notag \text{s.t.} \phantom{m} \sum_{X: u\in X} y_X+\alpha_i & \ge w(u)  
    & (u \in S_i,
       i\in [n]),\\
       \notag \alpha_i & \ge 0  & (i \subseteq [n]),\\
       \notag y_X & \ge 0  & (X \subseteq S).
\end{align}
\end{minipage}
\end{table}

It follows from the submodularity of the rank function $r$ 
that there exists a dual optimal solution $(\vec{y}, \vec{\alpha})$ 
such that the support $\C=\{\,X \subseteq S : y_X>0\,\}$ of $\vec{y}$ is a chain
 (see, e.g., \cite[Theorem~41.12]{Schrijver},\cite[Theorems~13.2.10 and 5.5.7]{Frankbook}). 
Let $(\vec{y}, \vec{\alpha})$ be such an optimal solution for \ref{LP2}. 
By the integrality of the matroid intersection polytope, \ref{LP1} admits integral optimal solutions, and hence the optimal value of \ref{LP1} is $\opt(w)$.
It then follows that 
a common independent set $I\in \cI_1\cap \cI_2$ is of maximum weight 
if and only if its characteristic vector satisfies the complementary slackness conditions with $(\vec{y}, \vec{\alpha})$, which is equivalent to the following claim.

\begin{claim}
\label{CLcs_one}
Let $(\vec{y}, \vec{\alpha})$ be an 
optimal solution for \ref{LP2}
such that the support $\C=\{X \subseteq S : y_X>0\}$ of $\vec{y}$ is a chain. 
A common independent set $I\in \cI_1\cap \cI_2$ satisfies $w(I)=\opt(w)$
if and only if the following three conditions are satisfied:
\begin{description}
\item[(1.1)] For any $i\in [n]$ and any $u\in I \cap S_i$, it holds that $\sum_{X: u\in X} y_X+\alpha_i = w(u)$.
\item[(1.2)] For any $i\in [n]$ with $\alpha_i>0$, it holds that $|I\cap S_i| = 1$.
%%\item[(1.3)] For any $C\in \C$, it holds that $|I\cap C|= \rank(C)$.
\item[(1.3)] For any $C\in \C$, it holds that $|I\cap C|= r(C)$.
\end{description}
\end{claim}
\begin{proof}
Conditions (1.2) and (1.3) are the complementary slackness conditions with respect to the feasibility constraints in \ref{LP1}, and (1.1) is the one with respect to the feasibility constraints in \ref{LP2}.
\end{proof}

This claim says that 
a popular maximum-weight common independent set is exactly a common independent set that is popular within those satisfying (1.1)--(1.3).
To reduce the popular maximum-weight common independent set problem 
to the popular common base problem, below we construct two matroids $M'_1$ and $M'_2$ 
such that 
the common bases of $M'_1$ and $M'_2$ correspond to the common independent sets in $\I_1\cap \I_2$ satisfying (1.1)--(1.3). 

One matroid $M'_1$ is defined in the following way. 
Let $T\subseteq S$ be the set of tight elements with respect to the constraints in \ref{LP2}, equivalently, elements satisfying the equation in (1.1):
\[T=\textstyle\bigcup_{i\in [n]}\left\{\, u\in S_i: \textstyle\sum_{X: u\in X} y_X+\alpha_i = w(u)\,\right\}.\]
For each $i\in [n]$, let $S'_i=S_i\cap T$ if $\alpha_i>0$ and $S'_i=(S_i\cap T)\cup \{d_i\}$ if $\alpha_i=0$, where $d_i$ is a dummy element not in $S$. 
%%and satisfies $v\succ_i d_i$ for every $v\in S_i\cap T$. 
Let $D$ be the set of the dummy elements, i.e., $D=\{\, d_i: i\in [n],\, \alpha_i=0\,\}$, 
and let $S'=T\cup D$. 
Note that $\{S'_1, S'_2,\dots, S'_n\}$ is a partition of $S'$. 
Now our matroid $M'_1=(S', \cI'_1)$ is a $1$-partition matroid on $S'$, 
defined from $\{S'_1, S'_2,\dots, S'_n\}$. 
Observe that, for any base $B$ of $M'_1$, 
a set $I\coloneqq B\cap S$ satisfies conditions (1.1) and (1.2), as well as $I \in \cI_1$.

\medskip
The other matroid $M'_2=(S', \cI'_2)$ is defined on $S'$ in the following manner. 
Let 
$|\C|=k$ 
and let 
$C_j$ be the $j$th inclusionwise minimal member of $\C$ for each $j\in [k]$, i.e., $\C=\{C_1, C_2,\dots, C_k\}$ and $C_1\subsetneq C_2\subseteq \cdots \subsetneq C_k$. 
%We can assume $C_1\neq \emptyset$, because changing the value of $y_{\emptyset}$ affects neither the objective value nor the feasibility in \ref{LP2}. 
Set $C_0=\emptyset$ and $C_{k+1}=S$. 
For each $j=1,2,..., k+1$, let $M_2^j$ be a matroid on $(C_j\setminus C_{j-1})\cap T$ obtained from $M_2$ by contracting $C_{j-1}$ and restricting to $(C_j\setminus C_{j-1})\cap T$. Also, let $N$ be the $(n-r(C_k))$-truncation of the direct sum of $M_2^{k+1}$ and the free matroid on $D$.
Now let $M'_2$ be the direct sum $M_2^1\oplus M_2^2\oplus \cdots\oplus M_2^k\oplus N$. 

Note that a maximum-weight common independent set $I_{\rm opt}$ satisfies (1.1) and (1.3), i.e., $I_{\rm opt}\subseteq T$ and $|I_{\rm opt}\cap C_j|=r(C_j)$ for each $j=1,2,\dots,k$. These imply that $r(C_j\cap T)=r(C_j)$ for each $C_j$, and hence the rank of each $M_2^j$ is $r(C_j)-r(C_{j-1})$. Then, for each $j=1,2,\dots, k$,  $M_2^1\oplus M_2^2\oplus \cdots\oplus M_2^j$ is a matroid on $C_j\cap T$ with rank $\sum_{k=1}^j(r(C_k)-r(C_{k-1}))=r(C_j)$.
Then, we can observe that, for any base $B$ of this new matroid $M'_2=M_2^1\oplus M_2^2\oplus \cdots\oplus M_2^k\oplus N$, a set $I\coloneqq B\cap S$ satisfies (1.1) and (1.3), as well as $I\in \cI_2$.

\smallskip
The correspondence between the common bases of $M'_1$ and $M'_2$ and 
the maximum-weight common independent set of $M_1$ and $M_2$ is observed as follows. 
For any common base $B$ of $M'_1$ and $M'_2$, the set $I\coloneqq B\cap S$ satisfies $I\in \cI_1\cap \cI_2$ and (1.1)--(1.3), i.e., $I$ is a maximum-weight common independent set. Conversely, for any maximum-weight common independent set $I$, a set $B_I=I\cup \{\, d_i: i\in [n],\, I\cap S_i=\emptyset\,\}$ is a common base of $M'_1$ and $M'_2$. 

The popularity in $\cI_1\cap \cI_2$ is also transferred to the common bases of $M_1'$ and $M_2'$. 
Recall that 
$u\succ_i \emptyset$ for each $i\in [n]$ and each $u\in S_i$. 
% in the partial order $\succ_i$ on $S_i \cup \emptyset$, 
% $\emptyset$ (assigned nothing) is worse than any element. 
We now construct a partial order $\succ_i'$ on $S_i'$ 
such that 
$u\succ_i' v$ if and only if $u\succ_i v$ for $u,v \in S_i\cap T$ 
and such that $v \succ_i' d_i$ for each $v \in S_i\cap T$. 
%%and that the dummy element $u_i$ is worse than any other element. 
Let $I, J\in \cI_1\cap \cI_2$ be maximum-weight common independent sets and $B_I, B_J$ be the corresponding common bases in $M'_1$ and $M'_2$.
It is straightforward to see that $\Delta(I, J)=\Delta(B_I, B_J)$. 
We thus conclude that $I \subseteq S$ is a popular maximum-weight common independent set in $\cI_1\cap \cI_2$  if and only if its corresponding common base $B_I$ is a popular common base in $M'_1$ and $M'_2$.

Therefore,  in order to find a popular maximum-weight common independent set in $M_1$ and $M_2$, 
it suffices to solve the popular common base problem for $M'_1$ and $M'_2$.
Since the above-mentioned dual optimal solution $(\vec{y}, \vec{\alpha})$ can be computed efficiently (see \cite[Theorems~13.2.10 and 5.5.7]{Frankbook}), we can construct $M'_1$ and $M'_2$ efficiently. 
It then follows from Theorem~\ref{thm:pop-base} that a popular common base for $M'_1$ and $M'_2$ can be computed in polynomial time. 
\end{proof}

%%\subsection{M$^\natural$-concave Utility Functions}
\subsection{Finding a Popular Maximum-Utility Common Independent Set}
\label{sec:Mconcave}
We prove Theorem \ref{thm:M-conc} by designing a polynomial reduction of the popular maximum-utility common independent set problem to the popular common base problem,  on the basis of Lemmas \ref{LEMargmax} and \ref{LEMsplit}. 

\begin{proof}[Proof of Theorem~\ref{thm:M-conc}]
Recall that our input consists of a $1$-partition matroid $M_1=(S, \cI_1)$ associated with partial orders $\{\succ_i\}_{i\in [n]}$, another matroid $M_2=(S, \cI_2)$, and an M$^\natural$-concave function $f\colon \cI_2 \to \R$.

Let $\opt(f)=\max\{\,f(I): I\in \cI_1\cap \cI_2\,\}$ and $\delta_{\cI_1}:\cI_1\to \R$ be a function 
on the independent set family $\cI_1$
that is constantly zero. 
Clearly $\delta_{\cI_1}$ is M$^\natural$-concave and 
$$\arg \max(\delta_{\cI_1}+f)=\{\, I\in \cI_1\cap \cI_2:  f(I)=\opt(f)\,\}.$$ 
We then apply Lemma~\ref{LEMsplit} with $\delta_{\cI_1}$ and $f$ in places of $f_1$ and $f_2$, respectively, to obtain that there exists a vector $p\in \R^S$ satisfying 
\[\arg \max(\delta_{\cI_1}+f)=\arg\max(\delta_{\cI_1}[+p])\cap \arg\max(f[-p]),\]
and that such $p$ can be computed efficiently. Note that the domains of $(\delta_{\cI_1}+f)$, $\delta_{\cI_1}[+p]$, and $f[-p]$ are $\cI_1\cap \cI_2$, $\cI_1$, and $\cI_2$, respectively. 
% By definition, 
% %%$\arg \max(\delta_{\cI_1}+f)=\{\, I\in \cI_1\cap \cI_2:  f(I)=\opt(I)\,\}$, 
% $\{\, I\in \cI_1\cap \cI_2:  f(I)=\opt(I)\,\}= \arg \max(\delta_{\cI_1}+f) 
% $, 

Now the popular maximum-utility common independent set problem is translated to 
the problem of finding a popular solution in the family $\arg\max(\delta_{\cI_1}[+p])\cap \arg\max(f[-p])$.
%\arg\max(\delta_{\cI_1}[+p])=\arg\max\{\, p(I): I\in \cI_1\,\}$, and $\arg\max(w[-p])=\arg\max\{\, w(I)-p(I): I\in \cI_2\,\}$
Below
we construct two matroids $M'_1$ and $M_2'$ so that 
$M_1'$ represents $\arg\max(\delta_{\cI_1}[+p])$ and 
$M_2'$ represents $\arg\max(f[-p])$.

%%Recall that $M_1=(S, \cI_1)$ is a partition matroid defined from $\{S_1, S_2,\dots, S_n\}$. 
The matroid $M_1'$ is constructed in the following way. 
For each index $i\in [n]$, let $p^{\star}_i=\max\{\,p(u): u\in S_i\,\}$ and 
$S^{\star}_i= \{\, u\in S_i: p(u)=p^{\star}_i\,\}$, 
%and $S_i^{\star}=\{u\in S_i: p(u)=p^{\star}_i\}$. 
and 
%%let $S'_i=S_i\cap T$ if $p^{\star}_i>0$ and $S'_i=(S_i\cap T)\cup \{d_i\}$ if $p^{\star}_i \le 0$, where $d_i$ is a dummy element not in $S$. Let $D$ be the set of the dummy elements, i.e., $D=\{\, d_i: i\in [n],\, p^{\star}_i\leq 0\,\}$, 
%%and let $S'=T\cup D$. 
define a set $S_i'$ by 
\begin{align*}
S_i' = 
\begin{cases}
    S_i^{\star} & (p_i^{\star} > 0), \\
    S_i^{\star} \cup \{d_i\} & (p_i^{\star} = 0), \\
    \{d_i\} & (p_i^{\star} < 0), 
\end{cases}
\end{align*}
where $d_i$ is a dummy element not in $S$. Let $D$ be the set of the dummy elements, i.e., $D=\{\, d_i: i\in [n],\, p^{\star}_i\leq 0\,\}$. 
Define a subset $S^{\star}\subseteq S$ by 
%%\[S^{\star}=\bigcup_{i\in[n]:p^{\star}_i\geq 0}\{\, u\in S_i: p(u)=p^{\star}_i\,\}.\] 
\[S^{\star}=\bigcup_{i\in[n]:p^{\star}_i\geq 0}S_i^{\star}\] 
and let $S'=S^{\star}\cup D$. 
Note that $\{S_1',S_2',\ldots, S_n'\}$ is a partition of $S'$. 
Now our matroid 
$M'_1=(S', \cI'_1)$ is the $1$-partition matroid on $S'$ defined by the partition $\{S'_1, S'_2,\dots, S'_n \}$. 
Observe that
%%the set $I\coloneqq B\cap S$ satisfies that $I\in \arg\max(\delta_{\cI_1}[+p])$.
\begin{align}
\label{EQm1}
\mbox{$B\cap S \in \arg\max(\delta_{\cI_1}[+p])$ for each base $B$ of $M'_1$}. 
\end{align}
Conversely, 
for any set  $I \in \arg\max(\delta_{\cI_1}[+p])$, 
it holds that 
$I\cup \{\, d_i: i\in [n], I\cap S_i=\emptyset\}$ is a base of $M_1'$.

%Let $M'_1=(S', \cI'_1)$ be the partition matroid defined from a partition $\{S'_1, S'_2,\dots, S'_r\}$, where each $S'_i$ is defined as follows. 
%We set $S'_i=S_i\cap T$ if $\alpha_i>0$ and otherwise set $S'_i=(S_i\cap T)\cup \{u_i\}$, where $u_i$ is a dummy element not in $S$ and satisfies $v\succ_i u_i$ for every $v\in S_i\cap T$. Let $D$ be the set of added dummy elements, i.e., $D=\{\, u_i: i\in [n],\, \alpha_i\leq 0\,\}$. Then, we see that $S'=T\cup D$. 
%We observe that, for any base $B$ of $M'_1$, the set $I\coloneqq B\cap S$ satisfies $I\in \arg\max(\delta_{\cI_1}[+p])$.

The matroid $M'_2=(S', \cI'_2)$ is defined in the following manner. 
First, define a family $\J \subseteq 2^{S^{\star}}$ by 
$$\J=\arg\max(f[-p])\cap 2^{S^{\star}},$$ i.e., $\J$ is a family of the sets of $\arg\max(f[-p])$ included in $S^{\star}$. 
It is derived from Lemma \ref{LEMargmax} that $\arg\max(f[-p]) \subseteq 2^S$ is an 
%%M$^\natural$-convex family on $S$ 
M$^\natural$-convex family, and hence its restriction $\J$ is an M$^\natural$-convex family on $S^{\star}$, which directly follows from the definition of M$^\natural$-convex families.
We then define a family $\B'_2\subseteq 2^{S'}$ by $$\B'_2=\{\,B\subseteq S': B\cap S^{\star}\in \J, |B|=n\,\}.$$ 
Namely, 
a set $B\in \B_2'$ is obtained from a set in $\J$ by adding some dummy elements in $D$ so that the resulting set has size $n$. 
It then follows from Lemma~\ref{lem:projection} that $\B'_2$ forms the base family of a matroid 
with ground set $S'$, 
%Take any $X, Y\in \B'_2$ and $x\in X\setminus Y$. Set $X_T\coloneqq X\cap T$ and $Y_T\coloneqq Y\cap T$. By definition, $X_T, Y_T\in \J$. If $x\not\in X_T$ and $|X_T|\geq |Y_T|$, which implies $x\in D$ and $Y\cap D\neq \emptyset$, we have $X-x+y\in \B'_2$ for any $y\in (Y\setminus X)\cap D$. If $x\not\in X_T$ and $|X_T|<|Y_T|$
and we define this matroid as $M_2'$. 
Namely, 
$M'_2=(S', \cI'_2)$, 
where 
$$\cI'_2=\{\, I\subseteq S': I\subseteq B\text{ for some } B\in \B'_2\,\}.$$ %%We observe 
It is straightforward to see that 
%%for any base $B$ of $M'_2$, the set $I\coloneqq B\cap S$ satisfies $I\in \arg\max(f[-p])\cap 2^T$.
% $B\cap S\in \arg\max(f[-p])$ 
% for any base $B$ of $M'_2$. 
\begin{align}
\label{EQm2}
\mbox{$B\cap S \in \arg\max(f[-p])$ for each base $B$ of $M'_2$}. 
\end{align}
It now follows from \eqref{EQm1} and \eqref{EQm2} that 
\begin{align*}
    \mbox{$B\cap S \in \arg\max(\delta_{\cI_1}[+p])\cap \arg\max(w[-p])$ for each common base $B$ of $M'_1$ and $M'_2$}. 
\end{align*}
%%Conversely, for any maximum utility common independent set $I$, a set $I\cup \{\, d_i: i\in [n], I\cap S_i=\emptyset\}$ is a common base of $M'_1$ and $M'_2$. 
As in the proof of Theorem \ref{thm:weighted}, 
for each agent $i\in [n]$, 
construct a partial order $\succ_i'$ on $S_i'$ 
such that $u\succ_i' v$ if and only if $u\succ_i v$ for $u,v \in S_i$ and such that $v \succ_i' d_i$ for each $v \in S_i$. 
%Let $I, J\in \cI_1\cap \cI_2$ be maximum weight common independent sets and $B_I, B_J$ be the corresponding common bases in $M'_1$ and $M'_2$. Recall that $\emptyset$ (assigned nothing) is worse than any element and that the dummy element $u_i$ is worse than any other element. Then, $\Delta(I, J)=\Delta(B_I, B_J)$. 
On the basis of the same argument as in the proof of Theorem~\ref{thm:weighted}, 
we conclude that 
$B \cap S$ is a popular maximum-utility common independent set 
for a popular common base $B$ in $M_1'$ and $M_2'$ with respect to the partial orders $\{\succ_i'\}_{i\in [n]}$. 
\end{proof}

\begin{remark}\label{rem:implementation}
Here we explain some implementation details of the maximum-utility common independent set algorithm shown in Theorem~\ref{thm:M-conc}. Note that we have membership oracles of $\cI_1$, $\cI_2$ and a value oracle of $f$, and we know $\emptyset\in \cI_1\cap \cI_2$.
Then, by Lemma~\ref{LEMsplit},
the vector $p$ used in the proof can be computed efficiently.
Since an M$^\natural$-concave function can be maximized efficiently, we can compute the value $\max(f[-p])$ and a maximizer in polynomial time. Then, a membership oracle of $\J=\arg\max(f[-p])\cap 2^{S^{\star}}$ is available and we can obtain some member $J\in \J$.
Then, the membership oracle of the base family $\B'_2=\{\,B\subseteq S': B\cap S^{\star}\in \J, |B|=n\,\}$ is also available and we can obtain some base $B\in \B'_2$.

Since it is known that a base oracle of a matroid together with a single known base is polynomially equivalent to an independence oracle \cite[p.37]{robinson1980computational}, \cite[p.175]{lovasz1982selected}, we can simulate a membership oracle of $\cI'_2$ using that of $\B'_2$.
\end{remark}

\section{Two-sided Preferences Model}
\label{sec:twosided}
%Our objective in this section is to prove Theorem \ref{thm:twosided} by designing a polynomial-time algorithm for finding a popular maximum-weight common independent set in the two-sided preferences model. 
In this section, we first provide a precise definition of popularity in the two-sided preferences model in Section~\ref{sec:prelimi2}. 
In Section~\ref{sec:reduction}, we show that our popular maximum-weight common independent set problem can be reduced to the popular critical common independent set problem, which is solved in Section~\ref{sec:critical}. This completes the proof of Theorem~\ref{thm:twosided}.

\subsection{Popularity in the Two-sided Preferences Model}
\label{sec:prelimi2}

%In this section, we describe the definition of popular common independent sets in the two-sided preferences model introduced by \cite{kamiyama2020popular}. 
%Preferences in this model are  total orders; recall that problems in two-sided preferences model become NP-hard when ties are introduced even in the simple bipartite matching setting \cite{biro2010popular}.

For clarity of presentation, we use the term `pairing' to mean a family of disjoint pairs of elements from two given disjoint subsets. That is, a \emph{pairing between $A$ and $B$} is a matching in the complete bipartite graph with vertex classes $A$ and $B$. 
%while a \emph{perfect pairing} is a perfect matching in the same graph. 

An ordered matroid $M$ is a tuple $(S, \cI, \succ)$, where $(S, \cI)$ is a matroid with $\cI$ being the independent set family and $\succ$ is a total order on $S$.
Let $M=(S, \cI, \succ)$ be an ordered matroid such that the matroid $(S,\cI)$ is given as a direct sum $M_1 \oplus M_2 \oplus \dots \oplus M_k$ for some positive integer $k$ and matroids $M_j=(S_j,\cI_j)$ ($j \in [k]$). 
Given an ordered pair of independent sets $(I,J)\in \cI\times \cI$, 
let $N$ be a pairing between $I \setminus J$ and $J \setminus I$. 
%and consider the following four conditions:
We say that $N$  is a {\em feasible pairing} for $(I,J)$ if the following conditions (FP1)--(FP4) hold.
\begin{enumerate}
\setlength{\leftskip}{3mm}
\item[(FP1)] $I-u+v \in \cI$ for every $uv \in N$, where  $u\in I\setminus J$ and $v\in J\setminus I$.
\item[(FP2)] %Any element $v\in J\setminus I$ with $I+v\not\in \cI$ is covered by $N$.
Any element $v\in J\setminus I$ that is uncovered by $N$ satisfies $I+v\in \cI$.
%\item[(3**)] Any element $u\in I\setminus J$ with $J+u\not\in \cI$ is covered by $N$.
\item[(FP3)] Every $uv\in N$ satisfies $u,v\in S_j$ for some $j \in [k]$.
\item[(FP4)] The number of pairs of $N$ induced by $S_j$ is $\min\{|S_j \cap (I \setminus J)|, |S_j \cap (J \setminus I)|\}$ for every $j\in [k]$. 
\end{enumerate}
We provide an explanation of these conditions taken from \cite{csaji2022solving}. As mentioned in Section~\ref{sec:contribution}, each summand $M_j$ of a matroid corresponds to an agent. Intuitively, conditions (1), (3) and (4) mean that the agent corresponding to $M_j$ compares $I$ and $J$ by pairing the elements of $S_j \cap (I\setminus J)$ to elements of $S_j\cap (J \setminus I)$ with which they can be exchanged, and comparing each pair. 
When $|S_j\cap (J \setminus I)|$ is larger than $|S_j \cap (I\setminus J)|$, some elements $v\in S_j \cap (J\setminus I)$ must be left unpaired. Such an element $v$ is regarded as being paired with $\emptyset$. Condition (2) requires that this kind of pair should also be exchangeable, i.e., $I-\emptyset+v=I+v\in \mathcal{I}$. A feasible pairing is known to exist.

\begin{lemma}[Kamiyama~\cite{kamiyama2020popular}]\label{lem:feasible}

For any $(I, J)\in \cI\times\cI$, there exists a feasible pairing for $(I,J)$. 
\end{lemma}

\noindent For independent sets $I$, $J$ and a feasible pairing $N$ for $(I,J)$, we define 
$\vote(I,J,N)\in \Z$ by 
\begin{align*}
    \vote(I,J,N)    =
    {}&{}|\{\,uv \in N: \mbox{$u\succ v$,  $u\in I\setminus J$, $v\in J\setminus I$}\}|
    \notag\\
    {}&{}-|\{\,uv \in N: \mbox{$u\prec v$, $u\in I\setminus J$, $v\in J\setminus I$}\}|+|I|-|J|. 
%\label{EQvoteIJN}
\end{align*}
Considering the most adversarial feasible pairing for $I$, we define $\vote(I,J) \in \Z$ as 
\begin{align*}
%\label{EQvoteIJ}
	\vote(I,J)&= \min\{\,\vote(I,J,N): N \text{~is a feasible pairing for $(I,J)$}\,\}. 
\end{align*}
Note that $\vote(I,J)$ is well-defined by Lemma~\ref{lem:feasible}.
We are now ready to describe popularity on matroid intersection. 

Let $M_1=(S, \cI_1, \succ_1)$ and $M_2=(S,\cI_2, \succ_2)$ be ordered matroids on the same ground set $S$. 
These matroids are 
given as direct sums $(S,\cI_1)= M^1_1 \oplus M^1_2 \oplus \dots \oplus M^1_{k_1}$ and $(S,\cI_2)= M^2_1 \oplus M^2_2 \oplus \dots \oplus M^2_{k_2}$. Each 
%summand $M^i_j=(S^i_j, \cI^i_j)$ 
matroid 
in the direct sums corresponds to an agent (voter), and hence there are $k_1+k_2$ agents. 
For each $i\in \{1,2\}$ and each ordered pair $(I,J)$ of common independent sets, 
we define $\vote_i(I,J)$ as above with respect to the ordered matroid $M_i$. 
We call a common independent set $I\in \cI_1\cap \cI_2$ {\em popular} if $\vote_1(I,J)+\vote_2(I,J) \geq 0$ for every common independent set $J\in \cI_1\cap \cI_2$. 
This definition of popularity is the same as the one in \cite{kamiyama2020popular, csaji2022solving}. See Remark~\ref{rem:def} for some discussions on other possible definitions.
%By using $\weakvote_i$ instead of $\vote_i$, we can define a stronger version of popularity, which we call {\em super popularity}. Namely, a common independent set $I\in \cI_1\cap \cI_2$ is \emph{super popular} if $\weakvote_1(I,J)+\weakvote_2(I,J) \ge 0$ for every common independent set $J\in \cI_1\cap \cI_2$. Note that this do not depend on the given decompositions of $M_1$ and $M_2$ into direct sums. Clearly from the definition, super popularity implies popularity.
It was shown in \cite{kamiyama2020popular} that a matroid kernel (defined below) is a popular common independent set, and hence a popular common independent set can be found efficiently using Fleiner's matroid kernel algorithm \cite{fleiner2001matroid,fleiner2003fixed}.

Let us now introduce a weight function $w\colon S \to \R$ and restrict our attention to maximum-weight common independent sets. 
%That is, it is of primary importance to maximize the total weight, and we aim to find a popular solution among common independent sets attaining maximum weight.
Let $\opt(w) = \max\{\,w(I) \colon I\in \cI_1\cap \cI_2\,\}$.
\begin{definition}
A common independent set $I \in \cI_1\cap \cI_2$ is called a {\em popular maximum-weight common independent set} if $w(I) = \opt(w)$ and 
$\vote_1(I,J)+\vote_2(I,J)\ge 0$ for each common independent set $J\in \cI_1\cap \cI_2$ with $w(J)=\opt(w)$. 
\end{definition}
Through Sections~\ref{sec:reduction} and \ref{sec:critical}, we show that there exists a polynomial-time algorithm that outputs a popular maximum-weight common independent set for any instance. This also serves as a proof of the existence of a popular maximum-weight common independent set, which is not obvious from the definition.

\paragraph{{\em Matroid Kernels.}}
We now formally describe our key tool, {\em matroid kernels} \cite{fleiner2001matroid, fleiner2003fixed}. It can be seen as a natural generalization of bipartite stable matchings. 
Let $M_1=(S, \cI_1, \succ_1)$ and $M_2=(S,\cI_2, \succ_2)$ be two ordered matroids on the same ground set $S$. For a common independent set $I \in \cI_1 \cap \cI_2$, we say that an element $v\in S\setminus I$ is {\em dominated} by $I$ in $M_i$ if $I+v \notin \cI_i$ and $u \succ_i v$ for every $u \in I$ for which $I-u+v \in \cI_i$. We call a common independent set $I \in \cI_1 \cap \cI_2$ an \emph{$(M_1,M_2)$-kernel} if every $v \in S \setminus I$ is dominated by $I$ in $M_1$ or $M_2$. If an element $v \in S \setminus I$ is dominated in neither $M_1$ nor $M_2$, we say that $v$ \emph{blocks} $I$.
Fleiner \cite{fleiner2001matroid,fleiner2003fixed} showed that a matroid generalization of the Gale-Shapley algorithm efficiently finds a matroid kernel. 

\begin{remark}\label{rem:def}
We provide some discussions on the definition of popularity. In contrast to the popularity in the one-to-one bipartite matching model, the concept of popularity is not so straightforward in the many-to-many matching model (with matroid constraints). 
%In this setting, each agent can be assigned multiple partners while she only has a preference on elements, rather than that on subsets, and hence there are  several different ways in which she can compare two matchings based on the sets of partners. 

The definition of popularity we are adopting is proposed in \cite{kamiyama2020popular} and used also in \cite{csaji2022solving}. This definition is reasonable in the sense that various interesting properties of popular matching in bipartite graphs extend to the matroid constrained setting under this definition. 
For example, one important fact on popularity is that it is a relaxation of stability, and this fact extends to the matroid constrained setting under the current definition of popularity. That is, a matroid kernel (which is arguably a natural matroid generalization of a stable matching) is a popular common independent set \cite{kamiyama2020popular}. Also, the tractability of the maximum popular matching problem extends to the matroid constrained setting~\cite{kamiyama2020popular, csaji2022solving} by generalizing the algorithm in the bipartite matching case \cite{huang2011popular, kavitha2014size} quite naturally.

In \cite{csaji2022solving}, some variants of popularity are investigated. The authors defined a {\em weakly feasible pairing} as a pairing that satisfies (FP1) and (FP2) (but not necessarily (FP3) and (FP4)) and defined a {\em super popularity} in the same manner as popularity by using weakly feasible pairings instead of feasible pairings. Super popularity is stronger than popularity and is independent from the direct sum representations of the two input matroids. Actually, we can observe from our proofs that the output of our algorithm is super popular. Therefore, we can find a super popular maximum-weight common independent set. We state our result Theorem~\ref{thm:twosided} with popularity (i.e., in a weaker form) as popularity has a more intuitive interpretation.

The authors of \cite{csaji2022solving} also proposed {\em defendability}: a common independent set $I$ is  {\em defendable} if $\vote_1(J,I)+\vote_2(J,I)\leq 0$ for every common independent set $J$. While the definition of popularity compares $I$ to $J$ using feasible pairings for $(I,J)$ that are most adversarial for $I$, the definition of defendability uses feasible parings for $(J,I)$ that is best possible for $I$. It was shown in \cite{csaji2022solving} that popularity implies defendability, which is not trivial for general matroids because feasible parings for $(I, J)$ are not the same as feasible pairings for $(J,I)$. Because the output of our algorithm showing Theorem~\ref{thm:twosided} is popular (moreover, super popular), it also satisfies defendability.

In \cite{csaji2022solving}, the authors also investigated other popularity notion, called {\em lexicographic popularity}, in which each agent casts only one vote comparing the sets assigned in two matchings lexicographically. For this definition of popularity, both existence and verification problems become coNP-hard even in the $b$-matching case.
%\todo{This remark is added to respond to reviewers' comments. If there is anything else that justifies the definition of our popularity, please add it.}
\end{remark}

\subsection{Reducing to the Popular Critical Common Independent Set Problem}
\label{sec:reduction}
As mentioned in Section~\ref{sec:contribution}, the popular maximum-weight common independent set problem in the two-sided preferences model cannot be reduced to previously solved unweighted problems.
%such as the popular common independent set problem and the popular common base problem \cite{kamiyama2020popular, Kuffner23}.
One reason is that the definitions of feasible pairings depend on the matroids, and hence the definition of popularity may be modified by some basic operations on matroids such as truncation and contraction. Another reason is that adding dummy elements causes comparisons between dummy elements in the reduced instance, which yield votes not corresponding to those in the original instance. 

We then introduce a new problem, the {\em popular critical common independent set} problem, and show that our problem can be reduced to it. 
%This problem can be seen as a generalization of the \emph{popular critical matching problem}  \cite{kavitha2021matchings}, and is described as follows. 

In the popular critical common independent set problem, 
we are given two ordered matroids $M_1=(S, \cI_1,\succ_1)$ and $M_2=(S, \cI_2, \succ_2)$, which are represented as directed sums as in Section~\ref{sec:prelimi2},
%%we are given two matroids $M_1=(S, \cI_1)$ and $M_2=(S, \cI_2)$, 
and two chains $\C_1, \C_2\subseteq 2^S$ on the ground set.  
%%where a chain $\C$ is a family of subsets such that for any distinct $C, C'\in \C$ we have $C\subsetneq C'$ or $C'\subsetneq C$. 
For each $i\in \{1,2\}$, the rank function of $M_i$ is denoted by $r_i$.
A common independent set $I\in \cI_1\cap \cI_2$ is called {\em $(\C_1, \C_2)$-critical} (or simply {\em critical}) if it satisfies $|I\cap C|=r_i(C)$ for any $i\in \{1,2\}$ and $C\in \C_i$. 
%We assume that the matroids admit a $(\C_1, \C_2)$-critical common independent set. 
\begin{definition}
For two chains $\C_1, \C_2\subseteq 2^S$,
a common independent set $I\in \cI_1\cap \cI_2$ is called a {\em popular critical common independent set} if $I$ is $(\C_1, \C_2)$-critical and satisfies $\vote_1(I,J)+\vote_2(I,J) \geq 0$ for every $(\C_1, \C_2)$-critical common independent set $J\in \cI_1\cap \cI_2$. 
\end{definition}
%We provide an efficient algorithm to solve this problem in Section~\ref{sec:critical}.
The {\em popular critical common independent set problem} asks to find a popular critical common independent set for given ordered matroids $M_1,M_2$ and chains $\C_1, \C_2\subseteq 2^S$, where we assume that they admit a $(\C_1, \C_2)$-critical common independent set.
To this problem, we reduce the popular maximum-weight common independent set problem. 

Consider the linear program \ref{LP3} below, corresponding to finding a maximum-weight common independent set $I$ in
$M_1$ and $M_2$, and its dual \ref{LP4}.
\vspace{-2mm}
\begin{table}[h]
\hspace{-5mm}
\begin{minipage}[t]{0.4\linewidth}\centering
\begin{align}
\displaystyle \mbox{Max.}& \quad \sum_{u \in S} w(u)\cdot x_u & \tag{LP3} \label{LP3}\\ 
\notag \text{s.t.}& \quad \sum_{u\in X}x_u \leq r_1(X) &  (X \subseteq S), \\ \notag
&\quad \sum_{u\in X} x_u  \leq r_2(X) & (X \subseteq S), \\ 
\notag
&\quad \qquad x_u \geq 0 & (u \in S).
\end{align}
\end{minipage}
%\hspace{0.4cm}
\hspace{0.05\textwidth}
\begin{minipage}[t]{0.53\linewidth}\centering
\begin{align}
\omit\rlap{$ \displaystyle \mbox{Min.} ~~\sum_{X \subseteq S }(y_X\cdot r_1(X)+z_X\cdot r_2(X))$} & & \tag{LP4} \label{LP4} \\ 
\notag \text{s.t.} \phantom{m} \sum_{X: u\in X} (y_X+z_X) & \ge w(u)  
    & (u \in S),\\
       \notag y_X & \ge 0  & (X\subseteq S), \\
       \notag z_X & \ge 0  & (X \subseteq S).
\end{align}
\end{minipage}
\end{table}

It is known that there exists an optimal solution $(\vec{y},\vec{z})$ for \ref{LP4} such that each of
the supports of $\vec{y}$ and $\vec{z}$ forms a chain (see, e.g., \cite[Theorem~41.12]{Schrijver}). 
Let $(\vec{y}, \vec{z})$ be such a solution and let $\C_1$ be the support of $\vec{y}$ and $\C_2$ be the support of $\vec{z}$. 
%We can assume that the minimal members of those chains are nonempty, because we can put $y_{\emptyset}=0$ and $z_{\emptyset}=0$ without affecting the objective value and feasibility of \ref{LP4}.
By the integrality of the matroid intersection polytope, \ref{LP3} admits integral optimal solutions, and hence the optimal value of \ref{LP3} is $\opt(w)$. 
We then derive the following claim from the complementary slackness of \ref{LP3} and \ref{LP4}.
% %For any common independent set $I\in \cI_1\cap \cI_2$, we have $w(I)=\opt(w)$ if and only if its characteristic vector satisfies complementary slackness conditions with $(\vec{y}, \vec{\alpha})$, which is equivalent to the following claim.

\begin{claim}
\label{CLcs_two}
Let $(\vec{y}, \vec{z})$ be an 
optimal solution for \ref{LP4}
such that the supports $\C_1=\{ X \subseteq S : y_X>0\}$ of $\vec{y}$ and 
$\C_2=\{ X \subseteq S : z_X>0\}$ of $\vec{z}$ are chains. 
A common independent set $I\in \cI_1\cap \cI_2$ satisfies $w(I)=\opt(w)$
if and only if the following conditions are satisfied:
\begin{description}
\item[(2.1)] For any element $u\in I$, we have that $\sum_{X: u\in X} (y_X+z_X) = w(u)$.
\item[(2.2)] $I$ is $(\C_1, \C_2)$-critical.
\end{description}
\end{claim}
\begin{proof}
Condition (2.2) is the complementary slackness conditions with respect to the feasibility constraints in \ref{LP3}, 
and (2.1) is those with respect to the feasibility constraints in \ref{LP4}.
\end{proof}

\begin{comment}
Using complementary slackness conditions, we get that an element $u\in S$ is not included in any maximum weight common independent set if $\sum_{X: u\in X} (y_X+z_X)>w(u)$. Set $T=$ and hence elements satisfying this inequality can be deleted from the ground set $S$. Therefore, from now on we suppose that any element $S$ is included in some \maxwind. Therefore, by complementary slackness, we get that a \cind\ $I$ is a \maxwind, if and only if 
\begin{enumerate}
    \item[(i)]\label{CS:matroid1} $|I\cap X|=r_1(X)$ for all $X\in \mathcal{C}_1$ and \label{slackness1}
    \item[(ii)]\label{CS:matroid2} $|I\cap X|=r_2(X)$ for all $X\in \mathcal{C}_2$. \label{slackness2}
\end{enumerate}
\end{comment}
Let $T\subseteq S$ be the set of the elements satisfying the equation in the condition (2.1), i.e.,
$T=\{\, u\in S: \textstyle\sum_{X: u\in X} (y_X+z_X) = w(u)\,\}$. 
Let $M^{\star}_1=(T, \cI_1^{\star},\succ_1^{\star})$ and $M^{\star}_2=(T,\cI_2^{\star},\succ_2^{\star})$ be the restrictions of $M_1$ and $M_2$ to $T$, i.e., 
for each $i\in\{1,2\}$, 
\begin{align*}
    \cI^{\star}_i=\{\, X: X\in \cI_i,\, X\subseteq T\,\}, 
    \qquad \mbox{$u \succ_i^{\star} v$ if and only if $u \succ_i v$} ~~ (u,v \in T). 
\end{align*}
For each $i\in \{1,2\}$, define a chain $\C^{\star}_i \subseteq 2^S$ by 
$\C^{\star}_i=\{\, C\cap T: C\in \C_i\,\}$. 
Note that a maximum-weight common independent set $I_{\rm opt}$ satisfies (2.1) and (2.2), i.e., $I_{\rm opt}\subseteq T$ and $|I_{\rm opt}\cap C|=r_i(C)$ for each $i\in \{1,2\}$ and $C\in \C_i$. These imply that $r_i(C\cap T)=r_i(C)$, and hence $(\C_1^{\star}, \C_2^{\star})$-criticality is equivalent to $(\C_1,\C_2)$-criticality for a subset of $T$.
It then follows from Claim \ref{CLcs_two} that a common independent set $I\in \cI_1\cap \cI_2$ satisfies $w(I)=\opt(w)$ if and only if $I$ is a $(\C^{\star}_1, \C^{\star}_2)$-critical common independent set in $M^{\star}_1$ and $M^{\star}_2$. 

Note that 
%%unlike truncation and contraction, 
restriction does not change the 
%%definition of popularity: 
popularity of a common independent set: 
for each $i\in \{1,2\}$ and each 
%%$I, J\in \cI'_i$, 
$I, J\in \cI_i^{\star}$, 
the set of feasible pairings for $(I,J)$ with respect to $M^{\star}_i$ coincides with that for $(I,J)$ with respect to $M_i$.
Therefore, 
we conclude that 
to find a popular maximum-weight common independent set for matroids $M_1$ and $M_2$, it is sufficient to find a popular $(\C^{\star}_1, \C^{\star}_2)$-critical common independent set for $M^{\star}_1$ and $M^{\star}_2$.

%%Thus, our problem is reduced to the popular critical common independent set problem.

\subsection{Popular Critical Common Independent Set Algorithm}\label{sec:critical}

Let $M_1=(S, \cI_1,\succ_1)$ and $M_2=(S, \cI_2, \succ_2)$ be ordered matroids, where the matroids are
given as direct sums as in Section~\ref{sec:prelimi2}.
%$(S,\cI_1)= M^1_1 \oplus M^1_2 \oplus \dots \oplus M^1_{k_1}$ and $(S,\cI_2)= M^2_1 \oplus M^2_2 \oplus \dots \oplus M^2_{k_2}$.
Let $\C_1,\C_2\subseteq 2^S$ be chains 
described as 
$\C_1=\{C^1_1,C^1_2\dots C^1_{d_1}\}$ and $\C_2=\{ C^2_1,C^2_2\dots , C^2_{d_2}\}$, 
where $C_1^1\subsetneq C_2^1\subsetneq \dots \subsetneq C_{d_1}^1$ and $C_1^2\subsetneq C^2_2\subsetneq \dots \subsetneq C_{d_2}^2$. 
%We can assume without loss of generality that $C^1_1$ and $C^2_1$ are nonempty because $|I\cap \emptyset|=r_i(\emptyset)$ holds for any $i\in\{1,2\}$ and $I\subseteq S$.
We assume the existence of a critical common independent set. Indeed, this is the case for the instance obtained through the reduction in Section~\ref{sec:reduction}.

Below we describe an algorithm to find a popular $(\C_1, \C_2)$-critical common independent set. In the algorithm, we perform two transformations of the matroids. We sometime abuse the notation $M_i$ to mean the matroid $(S, \cI_i)$ rather than the ordered matroid $(S, \cI_i, \succ_i)$.

\paragraph{First Transformation using the Chains.}
For each $i\in \{ 1,2\}$, from the original matroid $M_i$, we define a matroid $M'_i=(S, \cI'_i)$ as follows. 
Let $C^i_0=\emptyset$, $C_{d_i+1}^i=S$, and for each $j=1,2,\dots, d_i+1$, let $\hat{M}^i_j$ be a matroid on $C^i_j\setminus C_{j-1}^i$ obtained from $M_i$ by contracting $C_{j-1}^i$ and restricting to $(C^i_j\setminus C^i_{j-1})$.
Let $M'_i=(S,\cI_i')$ be the direct sum of these matroids, i.e., $M'_i=\hat{M}^i_1\oplus \hat{M}^i_2\oplus \dots\oplus \hat{M}^i_{d_i+1}$. Then, the following claim holds. 

\begin{lemma}\label{obs:maxweight-bases-conn} For a set $I\subseteq S$, the following two conditions are equivalent: 
\begin{itemize}
    \item[\rm (i)] $I$ is a $(\C_1, \C_2)$-critical common independent set of $M_1$ and $M_2$. 
    \item[\rm (ii)] $I$ is a common independent set of $M_1'$ and $M_2'$ satisfying $|I\cap C_{d_i}^i|=r_i(C_{d_i}^i)$ for $i\in \{1,2\}$.
\end{itemize}
\end{lemma}
\begin{proof}
It follows from the definitions of $M'_1$ and $M'_2$ that (i) implies (ii). To see the other direction, suppose that $I$ satisfies (ii). 
Let $i\in \{1,2\}$. 
Since $I$ is an independent set of $M'_i$, 
it holds that $I\in \cI_i$ and $|I\cap (C^i_j\setminus C^i_{j-1})| \leq r_i(C^i_j)-r_i(C^i_{j-1})$ for each $j\in [d_i]$. Then we obtain 
$$|I\cap C^i_{d_i}|=\sum_{j=1}^{d_i} |I\cap (C^i_j\setminus C^i_{j-1})|\leq \sum_{j=1}^{d_i} r_i(C^i_j)-r_i(C^i_{j-1})=r_i(C^i_{d_i})=|I\cap C^i_{d_i}|.$$ 
Hence, $|I\cap (C^i_j\setminus C^i_{j-1})| = r_i(C^i_j)-r_i(C^i_{j-1})$ for every $j\in [d_i]$. Therefore, $I$ satisfies (i).
\end{proof}

\newcommand{\cop}{\mathcal{K}}
\newcommand{\Cmax}{C_{\rm max}}
\newcommand{\rone}{\rho_1}
\newcommand{\rtwo}{\rho_2}
\newcommand{\ri}{\rho_i}

\paragraph{Second Transformation via Duplication.}
For each $i\in \{1,2\}$, from the matroid $M'_i$ obtained above and the original order $\succ_i$, 
we define an extended ordered matroid $M^{\star}_i=(S^{\star}, \cI^{\star}_i, \succ^{\star}_i)$ 
%%by replacing each element $u\in S$ with parallel copies 
in the following way. 
Let $\Cmax^i=C^i_{d_i}$ and $\ri =r_i(\Cmax^i)$ for each $i\in\{1,2\}$.
For each $u\in S$, 
first replace $u$ with a copy $u^0$.
If $u\in \Cmax^1$, then we add copies $u^1,\dots u^{\rone}$. If $u\in \Cmax^2$, then we add copies $u^{-1},\dots , u^{-\rtwo}$. (If $u\in \Cmax^1\cap \Cmax^2$, then $u$ gets $\rone+\rtwo+1$ copies). 
Denote the set of copies created for each $u\in S$ by $\cop (u)$, 
and let the extended ground set be $S^{\star}=\bigcup_{u\in S}\cop (u) $. 
A copy $u^k\in S^{\star}$ is called the {\em $k$-level copy} of $u$ and also called an {\em $k$-level element}.

For each $I^{\star} \subseteq S^{\star}$, 
define $\pi(I^{\star})\subseteq S$ by 
$\pi(I^{\star})=\{\,u\in S: I^{\star}\cap \cop (u) \neq \emptyset \,\}$. 
For each $i\in \{1,2\}$, the independent set family of $M^{\star}_i$ is defined by
\[\cI_i^{\star}=\{\,I^{\star} \subseteq S^{\star}: \pi(I^{\star})\in \cI_i',\  
\mbox{$|I^{\star}\cap \cop (u) |\leq 1$ for each $u \in S$}\,\}.\]

For each $i\in \{1,2\}$, 
we define 
a linear order $\succ^{\star}_i$ on $S^{\star}$ as follows. 
In $\succ^{\star}_1$, lower level elements are preferred; for any $k,l\in \{-\rho_2,\dots,-1,0,1,\dots,\rho_1\}$ with $k<l$, any $k$-level element is preferred over any $l$-level element, and the original preferences are preserved for the elements in the same level. Namely,  
$u^k\succ_1^{\star} v^l$ holds if and only if $k<l$, or $k=l$ and $u\succ_1 v$. 
In $\succ^{\star}_2$, the higher level elements are preferred; 
%when $i>j$, any $i$-level element is preferred over any $j$-level element, and the original preferences are preserved for the elements in the same level (i.e.,  $
$u^k\succ_2^{\star} v^l$ if and only if $k>l$, or $k=l$ and $u\succ_2 v$. We remark that the construction of $M^{\star}_i$ is a generalization of the ideas in the popular critical matching algorithm \cite{kavitha2021matchings}. 

The algorithm is now described as follows.

\begin{enumerate}
	\item Find an $(M^{\star}_1, M^{\star}_2)$-kernel $I^{\star}$.
            \label{ENUkernel}
	\item Output $I\coloneqq \pi(I^{\star})$.
\end{enumerate} 

Note that we can find a matroid kernel $I^{\star}$ in
Step \ref{ENUkernel} in $\mathcal{O}(|S^{\star}|^2)=\mathcal{O}(r^2|S|^2)$ time by Fleiner's algorithm \cite{fleiner2001matroid, fleiner2003fixed}. 

The output of the algorithm is a popular critical common independent set as stated in the following theorem, whose proof is given in the next subsection. 
By applying this algorithm to the instance obtained by the reduction in Section~\ref{sec:reduction}, we can find a popular maximum-weight common independent set. Thus, we complete the proof of Theorem~\ref{thm:twosided}.
\begin{theorem}\label{thm:twosided-alg}
The output $I$ of the algorithm is a popular $(\C_1, \C_2)$-critical common independent set of $M_1$ and $M_2$.
\end{theorem}

As the proof of Theorem~\ref{thm:twosided-alg} in the next subsection is technical, here we explain a few key points.
The second transformation in our algorithm is a natural matroid generalization of the transformation of the input bipartite graph in the popular critical matching algorithm \cite{kavitha2021matchings}.
Though we cannot directly extend the proof argument in \cite{kavitha2021matchings}, 
by appropriately integrating proof techniques from prior work \cite{kavitha2014size, csaji2022solving}, we can prove that the output $I$ of our algorithm satisfies $\vote_1(I, J, N_1) + \vote_2(I, J, N_2) \geq 0$ for any critical matching $J$ and weakly feasible pairings $N_1$ and $N_2$ with respect to the intermediate matroids $M_1'$ and $M_2'$, where a pairing is {\em weakly feasible} if it satisfies (FP1) and (FP2) (cf. Remark~\ref{rem:def}).
However, this alone is not sufficient for our objective, 
because we need to establish the above inequality for all feasible pairings with respect to the original matroids $M_1$ and $M_2$, 
rather than $M_1'$ and $M_2'$. 
%Indeed, the first transformation of the matroids in our algorithm has no corresponding part in the existing algorithms, making our analysis nontrivial.

A key observation to overcoming this issue is that any feasible pairing w.r.t.\ $M_1$ and $M_2$ is a weakly feasible pairing w.r.t.\ $M_1'$ and $M_2'$, under the assumption that $I$ and $J$ are critical (Claim~\ref{claim:exchanges} in the next subsection). This comes from the fact that, under this assumption, each element $C$ of the chain $\C_i~(i\in \{1,2\})$ is spanned by both $I$ and $J$, which forces any pair of a feasible pairing in $M_i$ to connect elements of the same depth within $\C_i$, 
guaranteeing that the pair is exchangeable also in $M'_i$. With this observation, we can conclude that the output $I$ is a popular critical common independent set with respect to $M_1$ and $M_2$.

\subsection{Proof of Theorem~\ref{thm:twosided-alg}}
We show that the output $I$ of the algorithm in Section~\ref{sec:critical} is a popular $(\C_1, \C_2)$-critical common independent set of $M_1$ and $M_2$. Before showing its popularity, we first show that the output $I$ is indeed critical.
For each element $u\in I$, its 
{\em level}  is defined as 
the level of the copy of $u$ in $I^{\star}$ 
and is denoted by $\lev(u)$, 
i.e., 
$\lev(u)=k$
when $I^{\star}\cap \cop (u)=\{u^k\}$. 

\begin{lemma}
\label{LEMcritical}
The output $I$ is a $(\C_1, \C_2)$-critical common independent set of $M_1$ and $M_2$.
\end{lemma}
\begin{proof}
On the basis of Lemma~\ref{obs:maxweight-bases-conn}, it suffices to show that $I\in \cI'_1\cap \cI'_2$ and $|I\cap \Cmax^i|=r_i(\Cmax^i)$ for each $i\in \{1,2\}$.
As $I^{\star}$ is a common independent set of $M^{\star}_1$ and $M^{\star}_2$, clearly $I\in \cI'_1\cap \cI'_2$. 
%, hence independent in $M_1$ and $M_2$ by $r_{M_i'}\le r_i$.
We then complete the proof by showing that $|I\cap \Cmax^i|=r_i(\Cmax^i)$ for $i\in\{1,2\}$. 
We only prove $|I\cap \Cmax^1|=r_1(\Cmax^1)$. Since the constructions of $M^{\star}_1$ and $M^{\star}_2$ are symmetric, the other equality is shown similarly.

Suppose to the contrary that $|I\cap \Cmax^1|<r_1(\Cmax^1)$. 
Take any $(\C_1, \C_2)$-critical common independent set $J$ of $M_1$ and $M_2$ (which exists by assumption).
%By Lemma~\ref{obs:maxweight-bases-conn}, $J$ is a common independent set of $M'_1$ and $M'_2$ satisfying $|I\cap \Cmax^i|=r_1(\Cmax^i)$ for each $i\in \{1,2\}$.
For the pair $(I,J)$, take feasible pairings $N_1$ and $N_2$ with respect to $M_1'$ and $M_2'$, respectively, where we see each matroid $M'_i$ as the direct sum $M'_i=\hat{M}^i_1\oplus \hat{M}^i_2\oplus \dots\oplus \hat{M}^i_{d_i+1}$ (rather than the given representation $M_i=M^i_1 \oplus M^i_2 \oplus \dots \oplus M^i_{k_i}$) in the conditions (FP3) and (FP4) of feasible pairings. The existence of a feasible pairing is guaranteed by Lemma~\ref{lem:feasible}. 
For each $i\in \{1,2\}$, by condition (FP3) of a feasible pairing, any pair $uv$ in $N_i$ satisfies either $u,v\in \Cmax^i$ or $u,v\in S\setminus \Cmax^i$. Also, by (FP4) and the criticality of $J$, all elements in $(I\setminus J)\cap \Cmax^i$ are paired (i.e., covered) by $N_i$.
As we have $|I\cap \Cmax^1|<r_1(\Cmax^1)=|J\cap \Cmax^1|$, there exists an element $v\in (J\setminus I)\cap \Cmax^1$ that is unpaired in $N_1$. Note that $v$ is paired with some element in $N_2$, since otherwise $I+v\in \cI'_1\cap \cI'_2$ by condition (FP2) and hence $I^{\star}+v^0\in \cI_1^{\star}\cap \cI^{\star}_2$, which contradicts the fact that $I^{\star}$ is a matroid kernel.

Consider a bipartite graph $G=(I\setminus J, J\setminus I; N_1\cup N_2)$, 
which is decomposed into alternating paths, cycles, and isolated vertices.
Since $v$ is covered only by $N_2$, there exists an alternating path $P$ that starts at $v$. Let $v_1=v$ and $v_1, u_1, v_2, u_2, \dots ,v_p,u_p$ be the elements on $P$ appearing in this order, where $v_q\in J\setminus I$ for each $q\in [p]$, $u_q\in I\setminus J$ for each $q\in [p-1]$, and $u_p$ is either $\emptyset$ or an element in $I\setminus J$. Then  $u_{q} v_{q}\in N_2$ and $u_{q} v_{q+1}\in N_1$ for each $q$.

As we have $v_1\in \Cmax^1$, there are copies $v_1^1,\dots, v_1^{\rone}$ of $v_1$ in the matroids $M^{\star}_1$ and $M^{\star}_2$. Since $v_1^{\rone}$ does not block $I^{\star}$ while $I^{\star}+v_1^{\rone}\in \cI^{\star}_1$, it must be dominated in $M^{\star}_2$, i.e., we must have $u_1^{\star}\succ^{\star}_2 v_1^{\rone}$ where $u_1^{\star}$ is the copy of $u_1$ in $I^{\star}$. Since larger indices are preferred in $\succ_2^{\star}$, we must have $\lev (u_1)=\rone$. This also implies $u_1\in \Cmax^1$ and hence $u_1$ is paired in $N_1$ with an element in $\Cmax^1$ (by the property of $N_1$ mentioned above). 
It then follows that $v_2\neq \emptyset$ and $v_2\in \Cmax^1$. 
Now $v_2$ has copies $v_2^1,\dots, v_2^{\rone}$. Since smaller indices are preferred in $\succ_1^{\star}$, we have $u^{\star}_1=u^{\rone}_1\not\succ^{\star}_1 v_2^{\rone-1}$. 
Then we must have $u_2^{\star}\succ^{\star}_2 v_2^{\rone-1}$, where $u_2^{\star}$ is the copy of $u_2$ in $I^{\star}$, and hence $\lev (u_2)\geq\rone-1$. If $\rone>1$, this implies $\lev(u_2)>0$, and hence $u_2\in \Cmax^1$. 
It is then derived from the same argument that $v_3\neq \emptyset$, $v_3\in \Cmax^1$, and 
$\lev (u_3)\geq\rone-2$.
Continuing this argument, for each $q\in [\rone-1]$, we obtain $v_{q} \neq \emptyset$, $v_{q}\in \Cmax^1$, and $\lev (u_{q+1})\ge \rone-q>0$, which implies $u_{q+1}\in \Cmax^1$. Hence, $\{u_1, u_2, \dots, u_{\rone}\}\subseteq \Cmax^1$, and this implies $|I\cap \Cmax^1|\ge \rone= r_1(\Cmax^1)$, a contradiction. 
\end{proof}

We are now ready to complete the proof of Theorem~\ref{thm:twosided-alg}.
Let $I$ be the output of the algorithm and $I^{\star}$ be the matroid kernel of $M^{\star}_1$ and $M^{\star}_2$ such that $I=\pi (I^{\star})$. 
Let $J$ be an arbitrary $(\C_1, \C_2)$-critical common independent set of $M_1$ and $M_2$.  
We show that $\vote_1(I,J)+\vote_2(I,J)\ge 0$. 
Let $N_1$ and $N_2$ be any feasible pairings for $(I,J)$ with respect to $M_1$ and $M_2$, respectively. It is sufficient to show $\vote_1(I,J, N_1)+\vote_2(I,J, N_2)\ge 0$.

\begin{claim}\label{claim:exchanges}
For each $i\in \{1,2\}$, $N_i$ satisfies (FP1) and (FP2) for $(I,J)$ with respect to $M'_i$, and includes a perfect matching between $(I\setminus J)\cap \Cmax^i$ and $(J\setminus I)\cap \Cmax^i$.
\end{claim}
\begin{proof}
Let $i\in \{1,2\}$. 
By conditions (FP1) and (FP2) of a feasible pairing, any $uv\in N_i$ satisfies $I-u+v\in \cI_i$ and an element $v\in J\setminus I$ is uncovered only if $I+v\in \cI_i$.
It follows from Lemma~\ref{LEMcritical} that $I$ is critical, and hence $|I\cap C^i_j|=r_i(C^i_j)$ for every $j\in [d_i]$. That is, $I\cap C^i_j$ spans $C^i_j$. This implies that, for every $v\in (J\setminus I)\cap C^i_j$, we have $I+v\not\in \cI_i$ and the fundamental circuit\footnote{For an independent set $I$ of a matroid and an element $v$ such that $I+v$ is dependent, it is known that $I+v$ contains a unique minimal dependent set, called the {\em fundamental circuit} of $v$ for $I$.} of $v$ for $I$ is included in $(I\cap C^i_j)\cup \{v\}$. Thus, for every $j\in [d_i]$, any element in $(J\setminus I)\cap C^i_j$ must be paired with an element in $(I\setminus J)\cap C^i_j$ in $N_i$. 

Since $I$ and $J$ are both critical, we have $|(I\setminus J)\cap C^i_j|=|(J\setminus I)\cap C^i_j|=r_i(C^i_j)$ for each $j\in [d_i]$. Then, the above property implies that $N_i$ includes a perfect matching between $(I\setminus J)\cap (C^i_j\setminus C^i_{j-1})$ and $(J\setminus I)\cap (C^i_j\setminus C^i_{j-1})$ for each $j\in [d_i]$, where $C_0=\emptyset$. Therefore, for any $uv\in N_i$, we have $I'\coloneqq I+u-v\in \cI_i$ and $|I'\cap C^i_j|=|I\cap C^i_j|=r_i(C^i_j)$ for every $j\in [d_i]$. This implies $I'\in \cI'_i$. 
Any element $v\in J\setminus I$ not paired in $N_i$ satisfies $I+v\in \cI_i$. As we have $|I\cap \Cmax^i|=r_i(\Cmax^i)$, this implies $v\not\in \Cmax^i$ and $I+v\in \cI'_i$. Thus, $N_i$ satisfies conditions (FP1) and (FP2) with respect to the matroid $M'_i$.
\end{proof}

It follows from Claim~\ref{claim:exchanges} that, for each $i\in \{1,2\}$, a feasible pairing $N_i$ for $(I,J)$ with respect to $M_i$ satisfies conditions (FP1) and (FP2) also with respect to the matroid $M_i'$ when $I$ is the algorithm's output and $J$ is any critical common independent set.
We proceed to show the required inequality $\vote_1 (I,J,N_1)+\vote_2 (I,J,N_2)\ge 0$.

Consider a bipartite graph $G=(I\setminus J, J\setminus I; N_1\cup N_2)$, which is decomposed into alternating paths, cycles, and isolated vertices. Note that any $v\in J\setminus I$ cannot be an isolated vertex since otherwise $v^0$ blocks $I^{\star}$. Hence, all isolated vertices are in $I\setminus J$. For each path and cycle $P$, define its \emph{score} as
\begin{align*}
\score(P)=&+|\{\,uv\in P :\,  uv\in N_i,\, u\succ_i v \text{ for some } i\in \{1,2\}\,\}| \\
&- |\{\,uv\in P :\,  uv\in N_i,\, u\prec_i v \text{ for some } i\in \{1,2\}\,\}|\\
&+2(|P\cap (I\setminus J)|-|P\cap (J\setminus I)|),
\end{align*}
where we assume $u\in I\setminus J$ and $v\in I\setminus J$ and identify $P$ with its edge set (resp., its vertex set) in the first and second terms (resp., in the third term).
Note that $\vote_1(I, J, N_1)+\vote_2(I,J,N_2)$ equals the sum of the scores of all cycles and paths in $G$ plus $2\cdot\#\text{\{isolated vertices of $I\setminus J$ in $G$\}}$.
Therefore, showing $\score(P)\geq 0$ for any path and cycle $P$ completes the proof of the inequality $\vote_1(I, J, N_1)+\vote_2(I,J,N_2)\geq 0$. 
%Also, showing $|P\cap (I\setminus J)|-|P\cap (J\setminus I)|\geq 0$ for any path/cycle $P$ completes the proof of the $|I|\geq |J|$.

Let $P$ be an alternating path or cycle and let $u_0, v_1, u_1, v_2, u_2, \dots ,v_p, u_p$ be the elements on $P$ appearing in this order, where $u_q\in I\setminus J$ and $v_q\in J\setminus I$ for each $q$, 
and we set $u_0=\emptyset$ if $P$ starts at $J\setminus I$, $u_p=\emptyset$ if $P$ ends at $J\setminus I$, and $u_0=u_p$ if $P$ is a cycle.
Without loss of generality, we assume $u_{q-1} v_q\in N_1$ and $u_{q} v_{q}\in N_2$ for each $q\in[p]$.
%For each $u_q$ with $u_q\neq \emptyset$, we let $\lev(u_q)$ be an integer $i$ such that $I^{\star}\cap \cop (u_q)=u_{q}^i$. 

\begin{claim}\label{claim:noncycle} 
%If $P$ is not a cycle, 
If $P$ is a path, 
then we have the following.
\begin{description}
  \item[\rm (i)]  If $u_0\ne \emptyset$, then $\lev (u_0)\ge 0$.
    \item[\rm (ii)] If $u_0=\emptyset$, then $\lev (u_1)\ge 0$. Also, if $\lev (u_1)=0$, then $u_1\succ_2 v_{1}$.
    \item[\rm (iii)] If $u_p\ne \emptyset$, then $\lev (u_{p})\le 0$.
    \item[\rm (iv)] If $u_p=\emptyset$, then $\lev (u_{p-1})\le 0$. Also, if $\lev (u_{p-1})=0$, then $u_{p-1}\succ_1v_p$.
  
\end{description}
\end{claim}
\begin{proof}
If $u_0\ne \emptyset$, then $u_0$ is not paired in $N_2$. From Claim \ref{claim:exchanges}, we obtain that $u_0\notin \Cmax^2$, and hence $\lev(u_0)\ge 0$. Thus, (i) is shown.
If $u_0=\emptyset$, then $v_1$ is not paired in $N_1$, which implies $I+v_1\in \cI'_1$ by (FP2) for $M_1'$, and hence $I^{\star}+v_1^0\in \cI_1^{\star}$. Since $I^{\star}$ is a matroid kernel (i.e., stable), then $v_1^0$ must be dominated in the matroid $M^{\star}_2$. 
Note that $u_1v_1\in N_2$ implies $I+v_1-u_1\in \cI'_2$ by (FP1) for $M'_2$, and hence $I^{\star}+v_1^0-u_1^{\star}\in \cI_2^{\star}$, where $u_1^{\star}\in I^{\star}$ is the $\lev(u_1)$-level copy of $u_1$. Then, we must have $u_1^{\star}\succ_2^{\star} v^0_1$ because $v^0_1$ is  dominated in $M^{\star}_2$.
As elements of higher levels are preferred to those of lower levels in  $\succ^{\star}_2$,
 we have $\lev (u_1)\ge 0$ and, if $\lev (u_1)=0$, then $u_1\succ_2v_1$ should hold. Thus, (ii) is shown.
 
We can show (iii) and (iv) analogously.
\end{proof}

\begin{claim}\label{claim:nonsteep}
For each $q\in [p]$ with $u_{q-1},u_q\neq \emptyset$, we have $\lev(u_{q})\geq \lev(u_{q-1})-1$ and the following.
\begin{description}
\item[\rm (a)] If $\lev(u_{q})=\lev(u_{q-1})$, then $u_{q-1} \succ_1 v_{q}$ or $u_{q} \succ_2 v_{q}$.
\item[\rm (b)] If $\lev(u_{q})=\lev(u_{q-1})-1$, then $u_{q-1} \succ_1 v_{q}$ and $u_{q} \succ_2 v_{q}$.
\end{description}
\end{claim}
\begin{proof}
Since $I^{\star}$ is stable, for every integer $i\in \{ -\rtwo,\dots,\rone \}$ such that $v_{q}^i\in \cop (v_{q})$, the copy $v_q^i$ should be dominated by $I^{\star}$ in $M_1^{\star}$ or $M_2^{\star}$. Note that, as we have $u_{q-1} v_q\in N_1$ and $u_{q} v_{q}\in N_2$, the condition (FP1) (with respect to $M'_1$ and $M'_2$) implies that  we have $I^{\star}+v_q^i-u_{q-1}^{\star}\in \cI^{\star}_1$ and $I^{\star}+v_q^i-u_q^{\star}\in \cI^{\star}_2$, where $u^{\star}\in I^{\star}$ is the $\lev(u)$-level copy of $u$. Therefore,  we must have at least one of $u^{\star}_{q-1}\succ_1^{\star} v_q^i$ and $u^{\star}_{q}\succ_2^{\star} v_q^i$. 

Suppose to the contrary that $\lev(u_{q})< \lev(u_{q-1})-1$. We now show the existence of an integer $i'$ that satisfies $\lev(u_{q})<i'<\lev(u_{q-1})$ and $v_{q}^{i'}\in \cop (v_{q})$, which implies $u^{\star}_{q-1}\not\succ_1^{\star} v_q^{i'}$ and $u^{\star}_{q}\not\succ_2^{\star} v_q^{i'}$ contradicting the stability of $I^{\star}$. 
Observe that $\lev(u_{q})< \lev(u_{q-1})-1$ implies that $\lev(u_{q-1})>0$ or $\lev(u_{q})<0$ holds.
If $\lev (u_{q -1})>0$, then $u_{q -1}\in \Cmax^1$, and hence $v_{q} \in \Cmax^1$ follows from Claim~\ref{claim:exchanges}. Thus, $i':= \lev (u_{q-1})-1\ge 0$ satisfies the required conditions.
If $\lev (u_{q})<0$, then $u_{q}\in \Cmax^2$, and hence $v_{q }\in \Cmax^2$ follows from Claim~\ref{claim:exchanges}. 
Thus, $i':=\lev (u_{q })+1\le 0$ satisfies the required conditions. 
Therefore, we have $\lev(u_{q})\geq \lev(u_{q-1})-1$.

From Claim~\ref{claim:exchanges}, we obtain that the $\lev(u_q)$- and $\lev(u_{q-1})$-level copies of $v_q$ belong to $\cop (v_{q})$.
Then, (a) and (b) follow from the fact that $u^{\star}_{q-1}\succ_1^{\star} v_q^i$ or $u^{\star}_{q}\succ_2^{\star} v_q^i$ must hold for every $i\in \{ -\rtwo,\dots, \rone\}$ with $v_{q}^i\in \cop (v_{q})$.
\end{proof}
For each $q\in [p]$, define $\score (u_{q -1}v_{q}u_{q})$ by
\begin{equation*}
\score (u_{q -1}v_{q}u_{q})=
\begin{cases}
2& \mbox{if $u_{q -1}\succ_1 v_{q}$ and $u_{q}\succ_2 v_{q}$},\\
0& \mbox{if either $u_{q -1}\succ_1 v_{q}$  or $u_{q}\succ_2 v_{q}$},\\
-2&\mbox{if neither $u_{q -1}\succ_1 v_{q}$ nor $u_{q}\succ_2 v_{q}$}.
\end{cases}
\end{equation*}
where $v_{q}\succ_i \emptyset$ always holds for any $i\in \{1,2\}$. 

Let $A=\{\, q\in [p]: \lev(u_{q})= \lev(u_{q-1})\}$, $B=\{\, q\in [p]: \lev(u_{q})= \lev(u_{q-1})-1\}$, and $C=\{\, q\in [p]: \lev(u_{q})>\lev(u_{q-1})\}$, where we let $\lev(u_0)=-\infty$ if $u_0=\emptyset$ and let $\lev(u_p)=+\infty$ if $u_p=\emptyset$.
It follows from Claim~\ref{claim:nonsteep} that $\{A, B, C\}$ is a partition of $[p]$. This claim also implies that $\score (u_{q -1}v_{q}u_{q})$ is at least $0$ if $q\in A$, is exactly $2$ if $q\in B$,  and is at least $-2$ if $q\in C$. 
If $P$ is a cycle, then $\score(P)=\sum_{q =1}^{p}\score (u_{q -1}v_{q}u_{q})$ and $|B|\geq |C|$. Thus, $\score(P)\geq 0$ immediately follows. We then assume that $P$ is a path. 
We consider the following four cases depending on whether $u_0$ and $u_p$ are $\emptyset$ or not.

%Recall that $\score(P)$ has the third term $2(|P\cap (I\setminus J)|-|P\cap (J\setminus I)|)$, which takes the value of $2$ or $0$ or $-2$ according to two or one or none of $u_0$, $u_k$ is $\emptyset$.
%as well as the first two terms that represent votes along edges on $P$. Then, we can obtain that 

If $u_0\ne \emptyset$ and $u_p\ne \emptyset$, then Claim~\ref{claim:noncycle} implies $\lev(u_0)\geq 0$ and $\lev(u_p)\leq 0$. Then $|B|\geq |C|$ and hence $\sum_{q =1}^{p}\score (u_{q -1}v_{q}u_{q}) \ge 0$. Since this coincides with the sum of the first two terms of $\score(P)$, i.e., 
\begin{align*}
+|\{\,uv\in P :\,  uv\in N_i,\, u\succ_i v \text{ for some } i\in \{1,2\}\,\}| \\- |\{\,uv\in P :\,  uv\in N_i,\, u\prec_i v \text{ for some } i\in \{1,2\}\,\}|
\end{align*}
and the third term of $\score(P)$ is $2(|P\cap (I\setminus J)|-|P\cap (J\setminus I)|)$, which is $2$, we obtain $\score(P)\geq 2>0$.

If $u_0=\emptyset$ and $u_p=\emptyset$, then Claim~\ref{claim:noncycle} implies $\lev(u_1)\geq 0$ and $\lev(u_{p-1})\leq 0$. 
We thus have $|B\cap \{2,3,\dots,p-1\}|\geq |C\cap\{2,3,\dots,p-1\}|$,
and hence $\sum_{q =2}^{p-1}\score (u_{q -1}v_{q}u_{q}) \ge 0$. Claim~\ref{claim:noncycle} also implies $u_1\succ_2 v_1$ and $u_{p-1}\succ_1 v_p$. These imply that the sum of the first two terms of $\score(P)$ is at least $2$, while the third term of $\score(P)$ is $-2$. Thus, $\score(P)\geq 0$.

If $u_0\neq\emptyset$ and $u_p=\emptyset$, then Claim~\ref{claim:noncycle} implies $\lev(u_0)\geq 0$ and $\lev(u_{p-1})\leq 0$. We thus have $\sum_{q =1}^{p-1}\score (u_{q -1}v_{q}u_{q}) \ge 0$. Claim~\ref{claim:noncycle} also implies $u_{p-1}\succ_1 v_p$. Then, the sum of the first two terms of $\score(P)$ is at least $1$, while the third term is $0$. Thus, $\score(P)\geq 1>0$.

Similarly, if $u_0=\emptyset$ and $u_k\neq\emptyset$, we obtain $\score(P)\geq 1>0$.

Therefore, in any case, $\score(P)\geq 0$ holds.
This completes the proof of $\vote_1 (I,J,N_1)+\vote_2 (I,J,N_2)\ge 0$.

Since this holds for arbitrary critical common independent set $J$ and arbitrary feasible pairings $N_1$ and $N_2$ for $(I,J)$, we conclude that $I$ is popular.

\section{Hardness Results on Popular Near-Maximum-Weight Matching Problems}\label{sec:hardness}
In this section, we show some hardness results on the problems of finding popular ``near-optimal'' solutions. We consider only problems on bipartite graphs, which clearly imply hardness results on the general matroid intersection settings. 

\subsection{One-sided Preferences Model (Proof of Theorem~\ref{thm:hardness-1-sided})}\label{sec:hardness1}
In this section, we present a result on the the following problem.
%\begin{problem}[\nearopt]
% Given a bipartite graph $G=(A,B;E)$ with weak preferences (i.e., preference lists with ties) $\succ_a$ for each $a\in A$, a weight function $w :E\to \mathbb{R}$, and a number $k$, determine the existence of a matching $M$ such that $w (M)\ge k$ and $\Delta (M,N)\ge 0$ for any matching $N$ with $w (N)\ge k$. In addition, return such $M$ if it exists.
%\end{problem}

\begin{problem}[\nearoptone]
Given a bipartite graph $G=(A,B;E)$ with weak preferences (i.e., preference lists with ties) $\{\succ_a\}_{a\in A}$, a weight function $w :E\to \{0,1\}$, and a number $k$, determine the existence of a matching $M$ such that $w (M)\ge k$ and $\Delta (M,N)\ge 0$ for any matching $N$ with $w (N)\ge k$. In addition, return such $M$ if it exists.
\end{problem}

While we do not know whether this problem is NP-hard, we show that it is at least as hard as the notoriously difficult \exm\ problem, for which no deterministic polynomial-time algorithm has been found since it was introduced by Papadimitriou and Yannakakis in 1982 \cite{papadimitriou1982complexity}.

\begin{problem}[\exm]
Given a bipartite graph $G=(A,B;E)$ with each edge colored red or blue and an integer $k$, determine the existence of a perfect matching $M$ with exactly $k$ red edges.
\end{problem}

Below is a restatement of Theorem~\ref{thm:hardness-1-sided}. We show this theorem in the rest of this section.
\begin{theorem}\label{thm:hardone}
A deterministic polynomial-time algorithm for \nearoptone\ implies a deterministic polynomial-time algorithm for \exm.
\end{theorem}
\begin{proof}

Let $I=(G,k)$ be an instance of \exm\ with $G=(A,B;E)$ being an edge colored graph. Let $|A|=|B|=n$ and suppose that the vertices in $A$ and $B$ are represented as $A=\{a_1, a_2,\dots, a_{n}\}$ and $B=\{b_1, b_2,\dots, b_{n}\}$. 

From $I$, we construct an instance $I'$ of \nearoptone, that consists of a bipartite graph $G'=(A',B';E')$, weak preference $\succ_a$ for each $a\in A'$, and a weight function $w:E'\to \{0,1\}$, and a number $k'$.

We first define the sets $A'$ and $B'$ in $G'$ as follows. For convenience, we call an element in $A'$ an agent and that in $B'$ an object. We denote by $d_i$ the degree of $a_i\in A$ in $G$.
%\begin{align*}
%A'=\{\, b'_i: b_i\in B\,\}\cup \{o_i^p: a_i\in A, p\in [d_i]\}\cup \{\,x_i^p: p\in [d_i-1]\,\}
%\end{align*}

\begin{itemize}
\item For each $b_i\in B$, we have an object $b'_i$ in $B'$.
\item For each $a_i\in A$, we have $2d_i$ agents $a_i^1,\dots, a_i^{d_i}, c_i^1,\dots, c_i^{d_i}$ in $A'$ and $2d_i-1$ objects $o_i^1,\dots, o_i^{d_i},x_i^1,\dots,x_i^{d_i-1}$ in $B'$.
\end{itemize}

Observe that $|A'|=\sum_{i=1}^n (2d_i)=2|E|$ and $|B'|=|B|+\sum_{i=1}^n (2d_i-1)=n+2|E|-n=2|E|$, i.e., we have $2|E|$ agents and $2|E|$ objects.
%Observe that the number of $a_i^\ell$ type agents is $\sum_{i=1}^n d_i=|E|$ and that of $c_i^\ell$ type agents is also $|E|$. So, there are $2|E|$ agents in $A'$.
We next define the edge set $E'$ and preferences. As $G'$ is constructed to be a simple graph, we describe preferences of agents as orders on adjacent objects, which are equivalent to orders on incident edges. For each agent $a_i\in A$, let $(a_i,b_{i(1)}),\dots , (a_i,b_{i(d_i)})$ denote the edges adjacent to $a_i$ in $G$.

\begin{itemize}
\item  
For each $a_i\in A$ and $\ell\in [d_i]$, the agent $a_i^\ell$ is adjacent to $o^\ell_i$ and $b'_{i(\ell)}$. Her preference is $o_i^{\ell}\succ b'_{i(\ell)}$ if $(a_i,b_{i(\ell)})$ is red, and $b'_{i(\ell)}\succ o_i^{\ell}$ if $(a_i,b_{i(\ell)})$ is blue.

\item For each $a_i\in A$ and $\ell\in [d_i]$, the agent $c_i^{\ell}\in A'$ is adjacent to $o_i^{\ell}$ and $x_i^p~(p\in [d_i-1])$. Her preference is $ (x_i^1\ x_i^2\cdots x_i^{d_i-1}) \succ o_i^{\ell}$. That is, $x_i^p~(p\in [d_i-1])$ are all tied and $o_i^{\ell}$ is worse than them.
\end{itemize}

Next, we define weights $w:E'\to \{0,1\}$. All edges of type $(c_i^{\ell},o_i^{\ell})$, $(c_i^{\ell},x_i^p)$ and $(a_i^{\ell},o_i^{\ell})$ have weight $1$. The weight of an edge of type $(a_i^{\ell},b'_{i(\ell)})$ is $1$ if it is red and $0$ if it is blue. Let the weight bound $k'$ of $I'$ be $2|E|+k-n$. This completes the construction of $I'$. We introduce the following notations for subsequent arguments.
\begin{align*}
E'_{\sf red}&=\{\, (a_i^\ell, b'_{i(\ell)})\in E': ~a_i\in A, \ell\in [d_i], (a_i, b_{i(\ell)}) \text{ is red in $G$}\,\},\\
E'_{\sf blue}&=\{\, (a_i^\ell, b'_{i(\ell)})\in E': ~a_i\in A, \ell\in [d_i], (a_i, b_{i(\ell)}) \text{ is blue in $G$}\,\},\\
\Mk&=\{\, N'\subseteq E':\text{ $N'$ is a matching in $G'$ satisfying $w(N')\geq k'$}\, \}.
\end{align*}
%We see that $w(e)=0$ for any $e\in E'_{\sf blue}$ and $w(e)=1$ for any $e\in E'\setminus E'_{\sf blue}$.
\begin{claim}\label{claim:wha1}
    Suppose that $I$ is a yes-instance of \exm, i.e. $G$ has a perfect matching with exactly $k$ red edges. Then, $I'$ is a yes-instance of \nearoptone, i.e., $G'$ has a matching $M\in \Mk$ such that $\Delta(M,N)\geq 0$ for any $N\in \Mk$. 
\end{claim}
\begin{proof}
Let $M$ be a perfect matching with exactly $k$ red edges. Create a matching $M'$ in $I'$ as follows. Initialize $M'$ with an empty set and do the following for each $a_i\in A$. 
%\begin{align*} M=&\{\,(a_i^{\ell}, b_i(\ell)), (c_i^{\ell}, o_i(\ell)): (a_i, b_{i(\ell)})\in M\,\}
%&\{\,(a_i^{\ell}, o_i^{\ell}), (c_i^{\ell}, x_i^{\ell})): (a_i, b_{i(\ell)})\in E\setminus M\,\}
%\end{align*}

\begin{itemize}
\item Let $\ell^*\in[d_i]$ be the number such that $(a_i,b_{i(\ell^*)})\in M$ (i.e., agent $a_i$ is assigned to the $\ell^*$th neighbor in $M$).
\item Add $(a_i^{\ell^*},b_{i(\ell^*)})$ and $(c_i^{\ell^*}, o_i^{\ell^*})$ to $M'$.
\item For each $\ell\in [d_i]\setminus\{\ell^*\}$, add $(a_i^{\ell},o_i^{\ell})$ to $M'$. 
\item Add disjoint $d_i-1$ pairs between $c_i^{\ell}~(\ell\in [d_i]\setminus \{\ell^*\})$ and $x_i^p~(p\in [d_i-1])$ to $M'$.
\end{itemize}

As $M$ is a perfect matching, the resultant $M'$ is a perfect matching in $G'$. In addition, as $M$ has $k$ red edges (i.e., $(n-k)$ blue edges), the weight of $M'$ is $\sum_{a_i\in A}2d_i -(n-k) =2|E|+k-n=k'$. Hence, $M'\in \Mk$.

Take any matching $N'\in \Mk$. We show $\Delta (M',N')\geq 0$. 
Note that the sum of the votes of $c_i^{\ell}$ type agents is always nonnegative (i.e., the number of $c_i^\ell$ agents who prefer $M'$ to $N'$ is no less than the number of $c_i^\ell$ agents who prefer $N'$ to $M'$), because their preferences are $(x_i^1\ x_i^2\cdots x_i^{d_i-1})\succ o_i^{\ell}$ and in $M'$ all $x_i^p$ objects are matched with $c_i^\ell$ agents. 
In the following, we show that the sum of the votes of $a_i^{\ell}$ type agents is also nonnegative, which completes the proof.

Recall that the edges in $\Eblue$ have weight $0$ and all other edges in $G'$ have weight $1$. 
Let $\gamma$ be the number of unmatched agents in $N'$. As $w(N')\ge k'$, at least $k'$ agents should be matched by weight $1$ edges in $N'$. Hence, at most $|A'|-k'-\gamma=2|E|-k'-\gamma=n-k-\gamma$ agents are matched by weight $0$ edges in $N'$.
That is, $|N'\cap \Eblue|\leq n-k-\gamma$.

Note also that the total number of $x_i^p$ and $o^\ell_i$ type objects is $\sum_{i=1}^n(2d_i-1)=2|E|-n$. Then, $w(N')\ge k'=2|E|+k-n$ implies that at least $k'-(2|E|-n)=k$ agents must be assigned to $b'_j$ type objects via weight $1$ edges. Therefore, $|N'\cap \Ered|\geq k$.

By the construction of $M'$, we have $|M'\cap \Ered|=k$ and $|M'\cap \Eblue|=n-k$. Therefore, we obtain $|N'\cap \Ered|\geq |M'\cap \Ered|$ and $|N'\cap \Eblue|\leq |M'\cap \Eblue|-\gamma$.

Recall that the preference of an agent $a_i^{\ell}$ is $o_i^{\ell}\succ b'_{i(\ell)}$ or $b'_{i(\ell)}\succ o_i^{\ell}$ depending on whether $(a^\ell_i,b'_{i(\ell)})$ belongs to $\Ered$ or $\Eblue$ and that $M'$ is perfect.
Therefore, she prefers $N'$ to $M'$ only if either of the following two holds: (i) $(a^\ell_i, b'_{i(\ell)})\in \Ered$, $N'(a^\ell_i)=o^\ell_i$, and $M'(a^\ell_i)=b'_{i(\ell)}$  or (ii) $(a^{\ell}_i, b'_{i(\ell)})\in \Eblue$, $N'(a^\ell_i)=b'_{i(\ell)}$, and $M'(a^\ell_i)=o^\ell_i$.

If there are $\eta_1$ agents to whom (i) applies, then $|N'\cap \Ered|\geq |M'\cap \Ered|$ implies that there are at least $\eta_1$ agents who are matched by $\Ered$ edges in $N'$ but not in $M'$, and hence they prefer $M'$ to $N'$.   
If there are $\eta_2$ agents $a_i^{\ell}$ to whom (ii) applies, then, as we have $|N'\cap \Eblue|\leq |M'\cap \Eblue|-\gamma$, there are at least $\eta_2 + \gamma$ agents, who are matched by $\Eblue$ edges in $M'$ but not in $N'$, and hence they prefer $M'$ to $N'$. 
Hence, the number of $a^\ell_i$ agents who prefer $M'$ to $N'$ is at least $\eta_1+\eta_2+\gamma$ while the number of agents who prefer $N'$ to $M'$ is $\eta_1+\eta_2$. Thus, the sum of the votes of $a_i^{\ell}$ type agents is nonnegative as required.
\end{proof}

\begin{claim}\label{claim:wha2}
Suppose that $I$ is a yes-instance of \exm. Then, any solution of \nearoptone, i.e., any popular solution in the set $\Mk$, has weight exactly $k'$ and is perfect. 
\end{claim}
\begin{proof}
Take any $N'\in \Mk$ that satisfies either of the following two: (a) $w(N')> k'$ or (b) $w(N')=k'$ and $N'$ is not perfect.
We show that $N'$ cannot be popular, which completes the proof. Since $I$ is a yes-instance, as shown in Claim \ref{claim:wha1}, we can construct a matching $M'$ that is popular in $\Mk$ and satisfies $w(M')=k'$.  It is sufficient to show $\Delta(M', N')>0$.

Let $\gamma$ be the number of unmatched agents in $N'$.  By a similar argument as in the proof of Claim~\ref{claim:wha1}, 
we have $|N'\cap \Ered|\geq k$ and $|N'\cap \Eblue|\leq n-k-\gamma$, where the strict inequalities hold if $w(N')>k'$. 
Then, in each of case (a) and (b), we have $|N'\cap \Eblue|< n-k$.  
Thus, we obtain $|N'\cap \Ered|\geq |M'\cap \Ered|$ and $|N'\cap \Eblue|\leq |M'\cap \Eblue|-1$. 

Consider $\Delta(M', N')$. By the same argument as in the proof of Claim~\ref{claim:wha1}, the sum of the votes of $c_i^{\ell}$ agents is nonnegative. 
%Note also that for each $i\in [n]$, $d_i-1$ agents in $c_i^{\ell}~(\ell\in [d_i])$ are assigned to their top choices by the construction of $M'$, and hence only the remaining one agent can prefer $N'$ to $M'$. 
Also, similarly to the proof of Claim~\ref{claim:wha1}, the above two inequalities imply that the number of $a^\ell_i$ agents who prefer $M'$ to $N'$ is strictly larger than  the number of agents who prefer $N'$ to $M'$. Therefore, $\Delta(M',N')>0$ is obtained. 
\end{proof}

Suppose that a matching $M'$ in $G'$ is perfect and satisfies $w(M')=k'$. 
Then, $w(M')=k'$ implies that $M'$ has $n-k$ edges in $\Eblue$, and then the perfectness implies that $M'$ has $k$ edges in $\Ered$. Hence the induced matching $M\coloneqq \{\, (a_i, b_{i(\ell)})\in E : i\in [n], \ell\in [d_i], (a^\ell_i, b'_{i(\ell)})\in M' \}$ is a perfect matching in $G$ with exactly $k$ red edges. 

This fact and Claims~\ref{claim:wha1} and \ref{claim:wha2} imply that if we have a polynomial-time algorithm for \nearoptone, then we can solve \exm\  in polynomial-time as follows. 
%If an instance $I$ of \exm\ admits a perfect matching with exactly $k$ red edges, then by a deterministic polynomial-time algorithm that solves \nearoptone\ can find a popular solution in $\Mk$ that must be perfect and has weight exactly $k'$, and yields a solution to $I$. 
%Therefore, we can decide whether $I$ is a yes-instance or not as follows. 
Given an instance $I$ of \exm, we first construct the corresponding instance $I'$ and run an algorithm for \nearoptone\ on it. If the algorithm concludes that $I'$ has no solution, then we conclude that $I$ is a no-instance of \exm. If the algorithm returns a solution $M'$ of $I'$, then we check whether the matching in $G$ induced from $M'$ is a solution of $I$. If so, then $I$ is a yes-instance and otherwise we conclude that $I$ is a no-instance. 
By Claims~\ref{claim:wha1} and \ref{claim:wha2}, if $I$ is a yes-instance, then this algorithm correctly finds a solution of $I$. 
If $I$ is a no-instance, then it has no solution, and hence the algorithm must conclude that it is a no-instance.
\end{proof}

\begin{remark}
We can show analogously that ``\exm-hardness'' remains to hold even for the case with strict preferences, if the objects can have nonnegative integer capacities. To see this, observe that the reduction in the proof of Theorem~\ref{thm:hardone} constructs an instance of \nearoptone\ in which all ties are of the form $(x_i^1\ x_i^2\cdots x_i^{d_i-1})$. Hence, by replacing $d_i-1$ objects $x_i^p~(p\in [d_i-1])$ with an object $x_i$ with capacity $d_i-1$, we can transform the instance to the one without ties and with capacities.  
\end{remark}

From a proof analogous to that of Theorem~\ref{thm:hardone}, it also follows that the popular near-optimal-matching problem  with a general weight function $w:E\to\R$, denoted as \nearopt, is at least as hard as \exmw\cite{el2023exact} described below, for which not even a randomized polynomial-time algorithm is known. %and is believed to be NP-hard.   

\begin{problem}[\exmw]
Given a bipartite graph $G=(A,B;E)$ with each edge colored red or blue, a weight function $w :E\to \mathbb{R}$, an integer $k$ and a number $W\in \mathbb{R}$, determine the existence of a perfect matching $M$ with exactly $k$ red edges and with weight at least $W$.
\end{problem}

\begin{theorem}
A deterministic polynomial-time algorithm for \nearopt\ implies a deterministic polynomial-time algorithm for \exmw.
\end{theorem}
\begin{proof}
    Take the edge-colored, weighted bipartite graph, where we want to find a perfect matching with exactly $k$ red edges of weight at least $W$. 
    We create a new weight function as follows. Let $R$ be larger than any weight. Then, we add $nR$ to the weight of each red edge and set the minimum required weight to be $W + knR$. Here, any perfect matching of weight at least $W+nkR$ has to contain at least $k$ red edges. 

    By using the same reduction as Theorem~\ref{thm:hardone}, we can show that if there is a perfect matching with exactly $k$ red edges and weight at least $W$, then it must be popular among the matchings with weight at least $W+nkR$ (only a matching corresponding to one with less red edges could dominate it) and reversely, in this case, any popular matching among the ones with weight at least $W+nkR$ must contain exactly $k$ red edges and give a matching of weight at least $W$. 
\end{proof}

\subsection{Two-sided Preferences Model (Proof of Theorem~\ref{thm:hardness-2-sided-allone})}

In this section, we show the NP-hardness of the following problem.

\begin{problem}[\kpm]
Given a bipartite graph $G=(U,W;E)$ with strict preferences $\{\succ_v\}_{v\in U\cup W}$ and a number $k$, determine the existence of a matching $M$ such that $|M|\ge k$ and $\Delta (M,N)\ge 0$ for any matching $N$ with $|N| \ge k$.
\end{problem}

We show the hardness of \kpm\ by a reduction from the following NP-hard problem.
\begin{problem}[\expm]
Given a bipartite graph $G=(U,W;E)$ with strict preferences $\{\succ_v\}_{
v\in U\cup W}$ that admits a complete popular matching and a popular matching of size at most $|U|-2$, determine the existence of a popular matching of size $|U|-1$.
\end{problem}

\begin{theorem}[Faenza--Kavitha--Powers--Zhang \cite{faenza2019popular}]
    \expm\ is NP-complete.
\end{theorem}

\begin{corollary}\label{cor:extragadget}
    There is a bipartite graph $H=(U,W;E)$ and strict preferences $\{\succ_v\}_{v\in U\cup W}$, such that $|U|=|W|=n/2$, there is a popular matching of size $|U|$, there is a popular matching of size at most $|U|-2$, but there is no popular matching of size $|U|-1$.
\end{corollary}

\newcommand{\specedge}{special edge}

First we introduce a notion of an $\ell$-\emph{\specedge}, for $\ell \in \mathbb{N}$. An $\ell$-\specedge\ $e(v_i,v_j)$ consists of a path between $v_i$ and $v_j$ with $2\ell$ inner vertices $v(e)_1^1,v(e)_1^2,\dots, v(e)_{\ell}^1,v(e)_{\ell}^2$ (in that order) such that for each $i\in [\ell]$, $v(e)_i^1$ and $v(e)_i^2$ consider each other best and their other neighbor second (see Figure~\ref{fig:noinstance} for example).
In an $\ell$-\specedge, let us call the vertices within the special edge the \emph{connector vertices} and the endpoints of the special edge the \emph{corner vertices.}

We say that $M$ is \kpop\, if $|M|\ge k$ and $\Delta (M,N)\ge 0$ for any matching $N$ with $|N|\ge k$.

 First of all, we can assume that each connector vertex is matched in any \kpop\ $M$, because they are the first choice of some other vertex, hence if they remain unmatched, then there is a matching of the same size that dominates $M$ if we let those agent switch. Hence, for each $\ell$-special edge $e(v_i,v_j)$, we have two possible configurations, either $(v(e)_h^1,v(e)_h^2)\in M$ for all $h\in [\ell]$, in which case we say that $e(v_i,v_j)$ is \emph{not included in $M$ (or just $e(v_i,v_j)\notin M$)} or $(v_i,v(e)_1^1),(v(e)_1^2,v(e)_2^1),\dots,(v(e)_{\ell}^2,v_j)\in M $, in which case we say that $e(v_i,v_j)$ is \emph{included in $M$ (or just $e(v_i,v_j)\in M$)}. When we refer to the \emph{addition or deletion of a special edge $e(v_i,v_j)$}, we mean a switch between these two possible configurations of $e(v_i,v_j)$.

We provide an instance of \expm\ that admits no solutions (see Figure~\ref{fig:noinstance}).
\begin{figure}[b!]
\begin{center}
\includegraphics[scale=0.45]{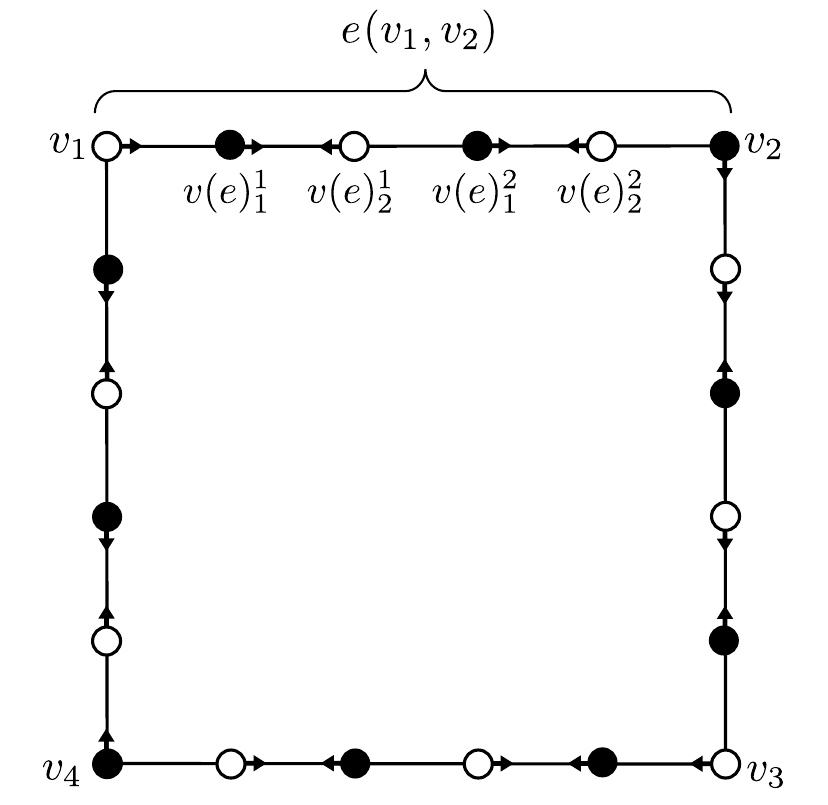}
\caption{\small An instance in Lemma~\ref{lem:cyclelemma} for the case $\ell=2$ and $K=2$. 
%It consists of four $\ell$-special edges. 
White and black colors represent the bipartition of the vertex set. A small arrow leaving from each vertex represents the edge that vertex most prefers. Four vertices $v_1, v_2, v_3, v_4$ are corner vertices and all other vertices are connector vertices.}
\label{fig:noinstance}
\end{center}
\end{figure}
\begin{lemma}\label{lem:cyclelemma}
    Let $C=\{ v_1,v_2, \dots, v_{2K}\}$ be a cycle of length $2K$, where each edge $(v_i,v_{i+1})$ is an $\ell$-\specedge\ (so altogether, $C$ has $(2\ell +1)2K$ vertices) and each $v_i$ prefers $v_{i+1}$ to $v_{i-1}$ (i.e. $v_i$ prefers the adjacent neighbor in the special edge $e(v_i,v_{i+1})$ to that in $e(v_{i-1}, v_i)$) for $i\in [2K]$. Let $2\ell K<k<(2\ell +1)K$ be an integer and suppose $\ell \ge 2$. Then, there is no \kpop\ in this instance.
\end{lemma}
\begin{proof}
Let $2\ell K<k<(2\ell +1)K$ and suppose that there is a \kpop\ $M$. 
Then, we know that $M$ includes more than $0$, but less than $K$ special edges. The latter follows from the fact that if $M$ contains $K$ special edges, then deleting one of them still gives a large enough matching, which then dominates $M$ ($2\ell\ge 4$ connector vertices improve and 2 corner vertices get worse). Let $i\in [K]$ be such that $e(v_i,v_{i+1})\in M$ and $e(v_{i+1},v_{i+2}), e(v_{i+2},v_{i+3})\notin M$. Then, create $N$ from $M$ by deleting the special edge $e(v_i,v_{i+1})$ and adding $e(v_{i+1},v_{i+2})$. Clearly, $|N|\ge k$. Consider $\Delta (M,N)$. Then, $4\ell$ connector vertices change partners, $2\ell$ of them votes with $-1$ and $2\ell$ of them with $+1$. Furthermore, among the corner vertices, $v_{i+1},v_{i+2}$ vote with $-1$ and only $v_i$ votes with $+1$. Hence $\Delta (M,N)<0$, a contradiction. 
\end{proof}

Below is a restatement of Theorem~\ref{thm:hardness-2-sided-allone}. We show this theorem in the rest of this section.
\begin{theorem}
    \kpm\ is NP-complete.
\end{theorem}
\begin{proof}
First we show that the problem is in NP. Let $M$ be an arbitrary matching. It is easy to check whether $|M|\ge k$. Deciding if there is a matching $N$ with $|N|\ge k$ and $\Delta (M,N)<0$ can be done as follows. First, let $E'$ be the original edge set $E$ extended by self-loops $(v,v)$ for each vertex $v$, i.e., $E'=E\cup \{\,(v,v): v\in U\cup W\}$. Define a cost function $c$ over $E'$ such that $c(e)=\vote_u(M,e)+\vote_w(M,e)$ for $e=(u,w)\in E$ and $c(v,v)=\vote_v(M,\emptyset)$ for a self-loop at $v$. Here, $\vote_u(M,e)$ is defined to be $+1$ if $M(u)\succ_u e$, $0$ if $M(u)=e$ and $-1$ otherwise. In this graph, the cost of any perfect matching $N'$ (i.e. it covers every vertex, but it can use the self-loops) is exactly $\Delta (M,N)$, where $N$ is obtained by deleting the self-loops from $N'$. Hence, verifying if $M$ is $k$-popular is equivalent to deciding if there is a perfect matching with negative cost that uses at least $k$ original edges. 

Consider the incidence matrix $A$ of the original bipartite graph. If we add an all $1$ row to the bottom of this matrix, then it remains a Totally Unimodular (TU) matrix (i.e. each subdeterminant is $0,\pm 1$). This fact follows from the characterization of Ghouila and Houri \cite{ghoulia1962characterisation}, which states that an integer matrix is TU, if and only if for any subset of the rows, there exists an equitable 2-coloring, meaning that we can partition these rows such that the sum of the elements in the two color classes differs by at most 1 in each column. For this matrix, if the last row is not included, then we can color the rows according to the two classes of the bipartite graph and if the last row is included, then we can color the last row blue and every other row red, the sum of the red rows will always be $0,1$ or $2$, as there are at most two $+1$-s in any column, and so this differs by at most $1$ from $1$. 

Consider the linear program
\begin{align*}
\mbox{Min.} \quad& \sum_{e\in E'} c(e)\cdot x(e) \quad \\ 
\notag \text{s.t.}
\quad &\sum_{e\in \delta_{E'}(v)}x(e)=1 \quad (v\in U\cup W), \\ 
&~~\sum_{e\in E}~x(e)\ge k,
\end{align*}
%$\min cx $ such that $\sum_{e\in d(v)}x(e)=1$ for all $e\in E'$ and $\sum_{e\in E}x(e)\ge k$. 
where $\delta_{E'}(v)$ is the set of edges in $E'$ incident to $v$ for each $v\in U\cup W$.
The constraint matrix of this is obtained from the above mentioned TU matrix by adding new columns for the self-loops with only one nonzero element, which is $+1$, so it remains TU. Therefore, as the bounding vector is integral, there exists an integer optimal solution. Finally, it is easy to see that an optimal integer solution is exactly a perfect matching of minimum cost containing at least $k$ original edges, so we can decide the existence of such a matching in polynomial time by solving the linear program. 

To show NP-hardness, we reduce from \expm. Let $I=(G_1=(U_1,W_1;E_1),\succ_1=\{\succ_v\}_{v\in U_1\cup W_1})$ be an instance of \expm. Let $(G_2,\succ_2)$ be a no-instance of \expm\ as in Corollary \ref{cor:extragadget}. Let $(G,\succ) = (G_1\cup G_2, \succ_1\cup \succ_2)$, where $G_1\cup G_2$ is just the disjoint union of $G_1$ and $G_2$. 

\begin{claim}\label{claim:keyclaim}
    $(G,\succ)$ is a yes-instance of \expm\ if and only if $(G_1,\succ_1)$ is a yes-instance of \expm. Furthermore, if $(G,\succ)$ is a no-instance of \expm, then for any matching $M$ of size $\frac{n}{2}-1 = \frac{|V(G)|}{2}-1$ there is a matching $N$ of size at most $\frac{n}{2}-2$ with $\Delta (M,N)<0$.
\end{claim}
\begin{proof}
    Let $M_1$ be a popular matching of size $\frac{n_1}{2}-1=\frac{|V(G_1)|}{2}-1$ in $G_1$. Extend $M_1$ with a complete popular matching of $G_2$. It is clear that this gives a popular matching in $G$ of size $\frac{n}{2}-1$. In the other direction, suppose that $M$ is a popular matching of size $\frac{n}{2}-1$ in $G$. If $M$ gives a complete matching in $G_1$, then it gives one with size $\frac{n_2}{2}-1$ in $G_2$, so $M$ is not popular. Otherwise, it gives a matching of size $\frac{n_1}{2}-1$ in $G_1$, which must be popular.

    Suppose that there is no popular matching of size $\frac{n}{2}-1$ in $G$ and let $M$ be a matching of size $\frac{n}{2}-1$. Then, $M_i=M\cap E[G_i]$ is not popular for some $i\in \{ 1,2\}$. Take a matching $N_i$ that dominates $M_i$ in $G_i$ and a minimum size popular matching $N_{3-i}$ in $G_{3-i}$. By the properties of $G_1,G_2$ this matching has size at most $\frac{n}{2}-2$ and dominates $M$ as desired.
\end{proof}

We proceed with the construction. We keep a copy of $(G\succ )$. Then, for each vertex $v\in V(G)$, we add a gadget $H_v$, that is a cycle of four $\ell$-\specedge s $e(v_1,v_2),e(v_2,v_3),e(v_3,v_4),e(v_4,v_1)$ with $\ell = 2$, such that each $v_i$ prefers the neighbor in the special edge $e(v_i,v_{i+1})$. Furthermore, we add an edge $(v,v_1)$ between $v$ and $H_v$ for each $v\in V(G)$ that is considered worst for $v$ and best for $v_1$.

 Let $K=100n$.
Finally, we add a disjoint cycle $C$ of $K$ $\ell$-\specedge s for $\ell =5$ with corner vertices $\{ u_1,\dots, u_K\}$ such that $u_i$ prefers the neighbor in the special edge $e(e_i,e_{i+1})$. Denote the instance obtained by $(G',\succ')$. Since the numbers of vertices in $G$, $\cup_{v\in V(G)}H_v$, and $C$ are $n$, $20n$, and $11K$, respectively, $G'$ has $11K+21n$ vertices.

Finally, let the threshold of the minimum size constraint for matchings be $k\coloneqq 5K+\frac{21n}{2}-1=510.5n-1$ (note that $n=V(G)$ is even as $|U|=|W|=n/2)$. In other words, at least $10K+21n-2$ vertices must be covered and at most $K+2$ vertices can remain uncovered.
This completes the construction of an \kpm\ instance. We define $\vote_v'(\cdot,\cdot)$ from $(G',\succ')$ in the same manner as before.

\medskip
In the rest, we show that $(G,\succ)$ is a yes-instance of \expm\ if $((G', \succ'),k)$ is a yes-instance of \kpm\ (Claim~\ref{clm:kpop_to_exact})  and that the other direction also holds (Claim~\ref{clm:exact_to_kpop}). Together with Claim~\ref{claim:keyclaim}, they complete the proof.
We start with some important observations. 
\begin{claim}\label{claim:properties}
Suppose that $M'$ is a \kpop\ in $(G',\succ ')$. Then, the following statements hold.
   \begin{enumerate}
       \item $M'$ contains no \specedge s of $C$.
       \item $M'$ induces a matching of size $\frac{n}{2}-1$ in $G$.
   \end{enumerate}
\end{claim}
\begin{proof}
1. Let us suppose that $M'$ contains a \specedge\ of $C$. Then, by Lemma \ref{lem:cyclelemma}, we have that $M'$ must contain exactly $K/2=50n$ special edges. Hence, the size of $M'$ is at least $5\times \frac{K}{2}+6\times \frac{K}{2}=550n$. Therefore, if we remove a special edge from $M'$ in $C$, the new matching $N'$ is still large enough, but $\Delta (M',N')=-8$, contradiction (10 connector vertices vote with $-1$ and only 2 corner vertices with $+1$).

2. By the first statement, we have that $M'$ contains no special edges from $C$, hence in $G'\setminus C$, at most 2 vertices can remain uncovered. As each $v\in V(G)$ is a best choice of some vertex, we can assume that all of them are covered. Suppose there is no uncovered vertex in $G'\setminus C$. Then, there must be a 2-\specedge\ included in $M'$ which we can delete and get a matching $N'$ still large enough, but with $\Delta (M',N')=-2$. Hence, there are exactly 2 uncovered vertices. 

Suppose that the two uncovered vertices are in the same gadget $H_v$. Then, $(v,v_1)\notin M$ and in the cycle $H_v$, only one 2-\specedge\ is included, contradiction to Lemma \ref{lem:cyclelemma}. 

Therefore, we get that the two uncovered vertices are in different $H_v$ gadgets. For each $v\in V(G)$ with $(v,v_1)\in M'$, there must be at least one uncovered vertex in $H_v$. Also, $H_v$ has exactly one uncovered vertex only if $(v,v_1)\in M'$. Then, we obtain that $M'$ gives a matching of size $\frac{n}{2}-1$ in $G$ as desired. 
\end{proof}

\begin{claim}\label{clm:kpop_to_exact}
    If there is a \kpop\ $M'$ in $(G',\succ')$, then there is a popular matching $M$ of size $\frac{n}{2}-1$ in $(G,\succ)$.
\end{claim}
\begin{proof}

Suppose that there is a \kpop\ $M'$. By Claim \ref{claim:properties}, this gives a matching of size $\frac{n}{2}-1$ in $G$. Suppose for the contrary that $M$ is not popular. Then, by Claim \ref{claim:keyclaim}, there is a matching $N$ that dominates $M$ with size at most $\frac{n}{2}-2$.

We create a matching $N'$ in $G'$ as follows. We first set $N'=(M'\setminus M)\cup N$ (which is not necessarily a matching at this moment). We have the following three cases for each $v\in V(G)$.

(i) If $v\in V(G)$ is uncovered in $N$, but was covered in $M$, we add $(v,v_1)$ to $N'$ and delete the \specedge\ in $M'$ adjacent to $v_1$. In this case, %the sum of votes (for $M'$ compared to $N'$) in $H_v$ is $-4$ 
$\sum_{v\in V(H_v)}\vote_v'(M',N')=-4$ and the number of uncovered vertices increases by $1$ (in $N'$ compared to $M'$).

(ii) If $v$ is covered in $N$, but it was uncovered in $M$, then we delete $(v,v_1)$ and include the \specedge\ in $H_v$ that was included in $M'$ and also another \specedge\ (to $v_2$ or $v_4$) which is now possible. In this case $\sum_{v\in V(H_v)}\vote_v'(M',N')=+4$ and the number of uncovered vertices decreases by $1$.

(iii) Otherwise, if $v$ is covered or uncovered in both $M$ and $N$, then we keep the edges of $M'$ in $H_v\cup \{ v\}$. In this case, the sum of votes in $H_v$ is $0$ and the number of uncovered vertices stays the same. 

Finally, we add some 5-\specedge s in $C$, such that $N'$ has size at least $k$.

As $N$ had size at most $\frac{n}{2}-2$, there are at least as many occurrences of case (i) as of case (ii). Let $f$ be the difference between them ($f$ is even). Then, we added $f/2$ \specedge s in $C$ to $N'$. Hence, the sum of votes in $V(G')\setminus V(G)$ is $-4f + 8f/2=0$.

Finally, it is easy to see that for $v\in V(G)$, $\vote_v(M,N)=\vote_v'(M',N')$, so $N'$ has size at least $k$ and dominates $M'$, contradiction.    
\end{proof}

\begin{claim}\label{clm:exact_to_kpop}
    If there is a popular matching $M$ of size $\frac{n}{2}-1$ in $(G,\succ )$, then there is a \kpop\ $M'$ in $(G',\succ')$.
\end{claim}
\begin{proof}
Let $M$ be a popular matching in $G$ of size $\frac{n}{2}-1$. We create a matching $M'$ as follows. We add no \specedge s in $C$. We add every edge of $M$. Then, if $v\in V(G)$ is covered in $M$, then we add \specedge s $e(v_1,v_2),e(v_3,v_4)$ and otherwise we add $(v,v_1)$ and \specedge\ $e(v_3,v_4)$. Then, $M'$ covers all but $K+2$ vertices, so it is large enough. 
We claim that $M'$ is a \kpop. 

Suppose for the contrary that some matching $N'$ with $|N'|\ge k$ dominates $M'$. 

First we show that we can assume that $N'$ covers all connector vertices. Let $N'$ be a matching with $\Delta (M',N')<0$ that covers the most connector vertices. Suppose for the contrary that some  connector vertex $v(e)_{i}^j$ is left unmatched in $N'$. Assume $j=1$ because the other case can be shown similarly. Consider a sequence $v(e)_i^1,v(e)_i^2,v(e)_{i+1}^1,v(e)_{i+1}^2,\dots$ of vertices
on this special edge and let $u$ be the first $v(e)_{h}^2$ vertex that is unmatched in $N'$. If there is no such a vertex, let $u$ be the terminal corner vertex. Let $P$ be the subpath of this special edge from $v(e)_i^1$ to $u$. Then, $P$ alternately uses edges not in $N'$ and those in $N'$. Denote by $V_P$ and $E_P$ the sets of vertices and edges on $P$, respectively.
In case $u$ is an unmatched connector vertex $v(e)_{h}^2$, $P$ is an augmenting path for $N'$ and the matching obtained as the symmetric difference $N''=N'\triangle E_P$ matches more connector vertices and satisfies $\vote_v'(M',N'')\le \vote_v'(M',N')$ for all vertices $v$, a contradiction. 
In case $u$ is the corner vertex, observe that $\sum_{v\in V_P}\vote_v'(M',N')\ge +1$, because either (i) $v(e)^1_i$ votes with $+1$ and all vertices in $V_P\setminus \{v(e)^1_i\}$ vote with $0$ or (ii) all vertices in $V_P\setminus \{u\}$ vote with $+1$. (Recall the two possible configurations of a special edge.) Let  $N''=N'\triangle E_P$. Then, 
$\vote_u'(M',N'')+\sum_{v\in V_P\setminus\{u\} }\vote_v'(M',N'')\le +1 + 0 \le \sum_{v\in V_P }\vote_v'(M',N')$. Hence, $\Delta (M',N'')\le \Delta (M',N')<0$, so $N''$ still dominates $M'$, but matches strictly more connector vertices, a contradiction.

As $M'$ covers every vertex that is not a corner vertex and has size exactly $k$, it must hold that $N'$ covers at least as many corner vertices as $M'$. 
We now claim that $\sum_{v\notin V(G)}\vote_v(M',N')\ge 0$ holds based on the following observations. 
\begin{enumerate}
\item 
In a gadget of an $M$-covered vertex $v$, if $N'$ has $l$ more uncovered corner vertices than $M'$ for some $l\in \{ 0,1,2,3,4\}$, then the sum of votes in $H_v$ is at least $-4l$.% (equality can only hold for $l=0,1$).

\item In a gadget of a non-$M$-covered vertex $v$, if $N'$ has $l$ more uncovered corner vertices than $M'$ for some $l\in \{ -1,0,1,2,3\}$, then the sum of votes in $H_v$ is at least $-4l$.

\item In the cycle $C$ (consisting of $K$ $5$-special edges), if $N'$ has $l$ more uncovered corner vertices than $M'$ for some $l\in \{ 0,-2,-4,\dots, -K\}$ (i.e., $N'$ contains $-l/2$ more $5$-special edges), then the sum of votes in $C$ is at least $-4l$.
\end{enumerate}

Combining these three observations with the fact that $N'$ covers at least as many corner vertices as $M'$, we get $\sum_{v\notin V(G)}\vote_v'(M',N')\ge 0$.

Let $N$ be the matching in $G$ induced from $N'$. For each vertex $v\in V(G)$ we have that $\vote_v(M,N)\le \vote_v'(M',N')$ (where the strict inequality holds only if neither $M$ nor $N$ covers $v\in V(G)$).

Hence, $0>\sum_{v\in V(G')}\vote_v'(M',N')\ge \sum_{v\in V(G)}\vote_v'(M',N')\ge \sum_{v\in V(G)}\vote_v(M,N)$, which contradicts the fact that $M$ is popular.  
\end{proof}

The theorem follows from Claims~\ref{claim:keyclaim},\ref{clm:kpop_to_exact}, and \ref{clm:exact_to_kpop}. 
\end{proof}

\section*{Acknowledgement}
We are grateful to  Telikepalli Kavitha for her helpful comments and providing us information on the problem solved in Theorem \ref{thm:hardness-2-sided-allone}. 
We thank anonymous reviewers for their helpful comments and suggestions.
The work was supported by the Lend\"ulet Programme of the Hungarian Academy of Sciences -- grant number LP2021-1/2021, by the Hungarian National Research, Development and Innovation Office -- NKFIH, grant number TKP2021-NKTA-62 and K143858, and by JST PRESTO Grant Number JPMJPR212B, JST ERATO Grant Number JPMJER2301, 
and  
JSPS KAKENHI Grant Numbers JP20K11699, JP24K02901, and JP24K14828. 
The first author was supported by the Ministry of Culture and Innovation of Hungary from the National Research, Development and Innovation fund, financed under the KDP-2023 funding scheme (grant number C2258525).

\bibliographystyle{plain}
\bibliography{sample}
\appendix
%\clearpage
\end{document}